\documentclass[onecolumn,12pt,draftclsnofoot]{IEEEtran}

\usepackage{mathrsfs}
\usepackage{mathptmx}
\usepackage{amsmath}
\usepackage{amsfonts}
\usepackage{amssymb}
\usepackage{graphicx}
\usepackage{xcolor}
\usepackage{float}
\usepackage{subfigure}
\usepackage{psfrag}
\usepackage{bm}
\usepackage{cite}

\newtheorem{defi}{Definition}
\newtheorem{thm}{Theorem}

\newtheorem{rem}{Remark}

\newtheorem{lem}{Lemma}
\def\tp{\mathrm{T}}
\def\Ds{\displaystyle}

\renewcommand\baselinestretch{1.2}

\begin{document}

\title{Control Analysis and Synthesis of Data-Driven Learning: A Kalman State-Space Approach}

\author{Deyuan Meng, {\it Senior Member}, {\it IEEE}
\thanks{This work was supported by the National Natural Science Foundation of China under Grant 61873013 and Grant 61922007.}
\thanks{The author is with the Seventh Research Division, Beihang University (BUAA), Beijing 100191, P. R. China, and also with the School of Automation Science and Electrical Engineering, Beihang University (BUAA), Beijing 100191, P. R. China (e-mail: dymeng@buaa.edu.cn).}
}

\date{}
\maketitle

\begin{abstract}
This paper aims to deal with the control analysis and synthesis problem of data-driven learning, regardless of unknown plant models and iteration-varying uncertainties. For the tracking of any desired target, a Kalman state-space approach is presented to transform it into two robust stability problems, which bridges a connection between data-driven control and model-based control. This approach also makes it possible to employ the extended state observer (ESO) in the design of data-driven learning to overcome the effect of iteration-varying uncertainties. It is shown that ESO-based data-driven learning ensures model-free systems to achieve the tracking of any desired target. In particular, our results apply to iterative learning control, which is verified by an example.
\end{abstract}

\begin{IEEEkeywords}
Data-driven learning, extended state observer, iterative learning control, model-free system, robust stability.
\end{IEEEkeywords}

\section{Introduction}\label{sec1}

\IEEEPARstart{D}{ata-driven} control has been considered an active area that plays a more and more important role in the modern engineering applications of, e.g., industrial process, power grid network, and transportation system. Different from traditional model-based control, data-driven control directly leverages the data information collected from the controlled systems, instead of resorting to the model information from them, to achieve the design for controllers and the analysis of system performances, such as stability (or convergence) and robustness \cite{hw:13}. It avoids building high-precision physical models of practical systems or depending upon complex physical models that are not tractable in designing controllers. As a consequence, the classical design and analysis methods of feedback controllers established in the Kalman state-space framework may no longer work effectively in the presence of data-driven controllers.

One of the desirable properties of the data-driven controllers is the learning ability that helps to refine high-precision system performances. Given any desired target, data-driven learning is focused on bettering controllers gradually via learning with the data information fed back from previous iterations (executions, trials, processes), by which it is thus able to realize the tracking of the desired target, regardless of unknown uncertainties from the plant model or external environment. This class of learning processes operates directly on the input and output data and, by contrast, requires less or no model knowledge. Consequently, it admits the possibility of constructing controllers of data-driven learning in an iterative manner to implement the high-precision tracking tasks and, simultaneously, to bypass the procedure of identifying accurate system models.

Of particular interest is data-driven learning concerned with repetitive tasks for systems operating over a fixed time interval. Every task repetition implemented for all time steps is referred to as an iteration, which yields the so-called ``iterative learning control (ILC)'' \cite{bta:06,acm:07}. By incorporating only input and output data information saved at previous iterations, ILC is capable of updating the control input at the current iteration to achieve the perfect output tracking of the desired target (or trajectory) over the fixed time interval with the evolution of iterations. Because repetition is the nature of the simple strategy for learning from the experience data through iterative updating, ILC is generally easy to implement and widely applies to many practical cases. Further, one of the main advantages of ILC is its natural ability to fully overcome the effect of repetitive (or iteration-invariant) uncertainties on the prescribed tracking tasks, regardless of the uncertainties in the plant model or from the environment owing to, e.g., external disturbances/noises and initial shifts \cite{mm:17}.

There have been proposed promising approaches to realizing data-driven learning, especially through combination with ILC. In \cite{chhj:15,clzhh:17,bhyy:20,mz:19}, an optimization-based adaptive design is introduced for the updating of learning laws, and a dynamical linearization tool is given for the estimation of unknown system parameters, with which a data-driven learning framework can be developed based on the input and output data, together with the estimated parameters. This class of optimization-based adaptive methods makes data-driven learning effectively applicable for nonlinear systems subject to unknown plant models. In \cite{ct:17}, a progressive updating strategy is presented for the design of learning filters, and thus an additional design degree of freedom can be utilized to accelerate the convergence of data-driven learning. A model predictive control design is integrated into data-driven learning in \cite{rb:18}, which achieves target-free learning for repetitive tasks. To address the effects of uncertainties on data-driven learning, an extended state observer (ESO)-based mechanism is adopted to overcome nonrepetitive uncertainties resulting from external disturbances and initial shifts during executing repetitive tasks in \cite{hzc:18,hchh:20}, and both renewal and recognition mechanisms are leveraged to accommodate multiple random constraints on data transmission in \cite{s:18}. There have also been proposed alternative ways to capture data-driven learning by model-free approaches (see, e.g., \cite{jps:11,bhc:13,rpppd:13,ro:19}).

In addition to the theoretical development achieved for data-driven learning, its applications have been extensively explored at the same time, see, e.g., \cite{hchh:20,jps:11} for linear motors, \cite{bhc:13} for farm vehicles, \cite{rpppd:13} for nonlinear servo systems, \cite{ro:19} for wide-format printers, and \cite{mmw:19,clh:20} for building room temperatures. Though both theories and applications have attracted more and more attention for data-driven learning, its control and analysis are far from complete, especially in comparison to the classical model-based control analysis methods exploited in the Kalman state-space framework \cite{am:06}. Of particular note is the fact about the independence of data-driven learning from the fundamental controllability/reachability and observability characteristics for the controlled systems (see also \cite{chhj:15,clzhh:17,bhyy:20,mz:19,ct:17,rb:18,hzc:18,hchh:20,s:18,jps:11,bhc:13,rpppd:13,ro:19,mmw:19,clh:20}). In contrast with this fact, some attempts have been made to study the controllability and/or observability issues of ILC \cite{hm:07,lgn:19} and of data-driven control \cite{wl:11,lyw:14}, respectively. Further attempts studying these issues have been developed to address stability, optimality, and robustness problems for data-driven control recently \cite{pt:20,wetc:20}, regardless of whether the data are guaranteed to be persistently exciting or not. However, how to develop a Kalman state-space framework for the control analysis and synthesis of data-driven learning is still an intriguing problem to be resolved. A feasible solution will be remarkable since the classical feedback-based tools will be applicable to data-driven learning.

To deal with the abovementioned problem, this paper targets at a class of input-output data sequences generated iteratively, and establishes a Kalman state-space framework for the control analysis and synthesis of robust data-driven learning regardless of iteration-varying uncertainties. This provides the possibility of integrating the state feedback approaches into the theoretical developments of data-driven learning that can be devoted to the high-precision tracking of any desired target. Our contributions specifically include the following aspects.
\begin{enumerate}
\item
Motivated by the standard control system theory \cite{am:06}, we propose a Kalman state-space description in the iteration domain for data-driven learning systems. This can ensure the feasibility of establishing a close connection between the properties of controllability and observability and the control analysis and synthesis of data-driven learning. As a consequence, we can leverage the design approaches of feedback-based and observer-based controllers to exploit data-driven learning laws.

\item
We introduce two concepts of robust $k$-stability for data-driven learning systems that aim at realizing the accurate tracking of the desired target in the presence of iteration-varying uncertainties. Different from the typical stability problem for robust control (see, e.g., \cite[p. 477]{gl:95}), robust $k$-stability problems involve two classes of properties: 1) boundedness to evaluate the overall system performance for all iterations and 2) attractiveness concerned with the asymptotic system behaviors. We give two attractiveness properties in accordance with the accurate tracking tasks of data-driven learning. Moreover, we provide necessary and sufficient conditions for data-driven learning, which are exploited by addressing the tracking problem through the state feedback stabilization.

\item
We incorporate the ESO-based controller design method into data-driven learning. By benefiting from ESO-based feedback design, we arrive at updating laws to overcome ill effect of iteration-varying uncertainties on data-driven learning tasks. We disclose that the introduction of ESO-based information can help data-driven learning to better its continuous dependence property upon the variation of iteration-varying uncertainties. Further, this robust result works effectively with respect to the model uncertainties in the data transfer matrix that maps input to output. It is worth particularly noting that we can realize the tracking of the desired target if the iteration-varying uncertainties have a zero-convergence variation rate, rather than being strictly iteration-invariant or having finite energies (thus, necessarily converging to zero along the iteration axis).
\end{enumerate}

In addition, a tree-step problem-solving strategy is leveraged to implement data-driven learning that works effectively, even without any prior knowledge of the plant model for generating the input and output data. This strategy bridges the connection between data-driven control and model-based control, and also provides a perspective for how to leverage model-based control tools to cope with the analysis and synthesis problems of data-driven control. An example concerned with ILC tracking tasks is included to illustrate the validity of our proposed data-driven learning method.

The remainder of this paper is organized as follows. A data-driven learning problem is introduced in Section \ref{sec2}, for which a nominal Kalman state-space framework leveraging ESO-based design is established in Section \ref{sec3}. The robust design of ESO-based data-driven learning under model uncertainties is studied in Section \ref{sec4}, which is further explored for model-free systems in Section \ref{sec5}. The application to ILC and an example are given in Section \ref{sec6}, and the conclusions are made in Section \ref{sec7}.

{\it Notations:} Let $\mathbb{Z}_{+}$ and $\mathbb{Z}$ denote the sets of nonnegative and positive integers, respectively; and $I$ and $0$ the identity and null matrices with appropriate dimensions, respectively. We define a forward difference operator of any sequence $\left\{f_{k}:k\in\mathbb{Z}_{+}\right\}$ as $\Delta:f_{k}\to\Delta f_{k}=f_{k+1}-f_{k}$, and hence denote
$\Delta^{i}f_{k}=\Delta\left(\Delta^{i-1}f_{k}\right)$, $\forall i\in\mathbb{Z}$, where $\Delta^{0}=\bf{I}$ becomes the identity operator. Let $\Delta$ also work for a parameterized sequence $\left\{f_{k}(t):k\in\mathbb{Z}_{+}\right\}$ in the face of any interested parameter $t\in\mathbb{Z}_{+}$, that is, $\Delta:f_{k}(t)\to\Delta f_{k}(t)=f_{k+1}(t)-f_{k}(t)$. For any matrix $H\in\mathbb{R}^{n\times m}$, $\left\|H\right\|$ is an induced matrix norm from some compatible vector norm (see, e.g., \cite[Chapter 5]{hj:85}). Particularly, for any square matrix $H$ (i.e., $m=n$), $\rho\left(H\right)$ is the spectral radius of matrix $H$, and $H>0$ represents a positive-definite matrix, in which a star $(\star)$ is often adopted to denote a term induced by the symmetry.

\section{Problem Statement}\label{sec2}

\subsection{Data Transfer Plant}

Consider a data sequence, with respect to the iteration index $k$, generated by
\begin{equation}\label{eq01}
\bm{Y}_{k}=P\bm{U}_{k}+\bm{N}_{k},\quad\forall k\in\mathbb{Z}_{+}
\end{equation}

\noindent where $\bm{Y}_{k}\in\mathbb{R}^{p}$, $\bm{U}_{k}\in\mathbb{R}^{m}$, and $\bm{N}_{k}\in\mathbb{R}^{p}$ are the data of output, input, and unknown noise or external disturbance, respectively; and $P\in\mathbb{R}^{p\times m}$ is the unknown data transfer matrix. It is worth highlighting that the unknown uncertainty involved in the plant (\ref{eq01}) comes in two different forms of
\[\aligned
P&:&&\hbox{unknown plant model}\\
\bm{N}_{k}&:&&\hbox{unknown iteration-varying uncertainty}.
\endaligned
\]

To overcome the control difficulty caused by the fact that we have no prior knowledge of $P$, we use a step-by-step problem-solving strategy, for which we without loss of generality define $P$ in the form of $P=P_{0}+P_{\delta}$, with a nominal model $P_{0}\in\mathbb{R}^{p\times m}$ and an unknown model uncertainty $P_{\delta}\in\mathbb{R}^{p\times m}$. Three steps are thus involved to achieve the controller design of the plant (\ref{eq01}):
\begin{enumerate}
\item[(S1)]
nominal design without model uncertainties, i.e., $P_{\delta}=0$; 

\item[(S2)]
robust design in the presence of model uncertainties, i.e., $P_{0}\neq0$ and $P_{\delta}\neq0$;

\item[(S3)]
model-free design without plant knowledge, i.e., $P_{0}=0$. 
\end{enumerate}

\noindent Clearly, we may perform control tasks for the plant (\ref{eq01}) through the three steps even without using any model knowledge of $P$. We will consider (\ref{eq01}) a bounded plant, where $\left\|P_{\delta}\right\|\leq\beta_{\delta}$ holds for some bound $\beta_{\delta}\geq0$ (thus, $\left\|P\right\|\leq\left\|P_{0}\right\|+\beta_{\delta}$).

For the measure of the variation of $\bm{N}_{k}$ versus the iteration $k$, we define
\[
\beta_{\Delta^{i}\bm{N}}=\sup_{k\in\mathbb{Z}_{+}}\left\|\Delta^{i}\bm{N}_{k}\right\|,\quad
\beta_{\Delta^{i}\bm{N}}^{ess}=\limsup_{k\to\infty}\left\|\Delta^{i}\bm{N}_{k}\right\|,\quad\forall i\in\mathbb{Z}_{+}.
\]

\noindent We will consider $\bm{N}_{k}$ to be bounded without loss of generality, i.e., $\beta_{\bm{N}}<\infty$. We can validate that for any bounded $\bm{N}_{k}$, $\Delta^{i}\bm{N}_{k}$ is also bounded for each iteration $k\in\mathbb{Z}_{+}$, and
\[0\leq\beta_{\Delta^{i}\bm{N}}^{ess}\leq\beta_{\Delta^{i}\bm{N}}\leq2^{i}\beta_{\bm{N}},\quad\forall i\in\mathbb{Z}_{+}.
\]

\noindent Furthermore, the variation of the iteration-varying uncertainty $\bm{N}_{k}$ is said to disappear (respectively, quasi-disappear) if $\beta_{\Delta\bm{N}}^{ess}=0$ (respectively, $\beta_{\Delta^{2}\bm{N}}^{ess}=0$).

\begin{rem}\label{rem01}
For convenience, we directly call $\Delta\bm{N}_{k}$ and $\Delta^{2}\bm{N}_{k}$ the variation and the variation rate of $\bm{N}_{k}$, respectively. We can verify that for any $i,j\in\mathbb{Z}_{+}$ and $i<j$,
\[
\beta_{\Delta^{i}\bm{N}}^{ess}=0
\Rightarrow
\beta_{\Delta^{j}\bm{N}}^{ess}=0
\quad\hbox{but}\quad
\beta_{\Delta^{i}\bm{N}}^{ess}=0
\nLeftarrow
\beta_{\Delta^{j}\bm{N}}^{ess}=0.
\]

\noindent In particular, for $i=1$ and $j=2$, this implies that the variation of the iteration-varying uncertainty quasi-disappears provided it disappears, whereas the opposite may not be true. A similar fact worth noting is that for $\bm{N}_{k}$ with its variation disappearing, $\bm{N}_{k}$ may not need to converge. Two common considered trivial cases of $\bm{N}_{k}$ whose variation disappears are $\lim_{k\to\infty}\bm{N}_{k}=\bm{N}$ and $\bm{N}_{k}\equiv\bm{N}$, $\forall k\in\mathbb{Z}_{+}$ for some iteration-invariant vector $\bm{N}\in\mathbb{R}^{p}$.
\end{rem}
%

\subsection{Output Tracking Problem}

{\bf Problem Statement.} For the desired output target $\bm{Y}_{d}\in\mathbb{R}^{p}$, the problem addressed in this paper is to find an updating law of the input $\bm{U}_{k}$ such that the output $\bm{Y}_{k}$ of the plant (\ref{eq01}) achieves the tracking of $\bm{Y}_{d}$ as accurately as possible with the increasing of the iteration $k$.

To more precisely characterize the abovementioned tracking problem, we define the tracking error as
\[
\bm{E}_{k}=\bm{Y}_{d}-\bm{Y}_{k},\quad\forall k\in\mathbb{Z}_{+}.
\]

\noindent Then we can combine (\ref{eq01}) to arrive at
\begin{equation}\label{eq02}
\bm{E}_{k+1}
=\bm{E}_{k}-\Delta\bm{Y}_{k}
=\bm{E}_{k}+P\overline{\bm{U}}_{k}+\bm{D}_{k},\quad\forall k\in\mathbb{Z}_{+}
\end{equation}

\noindent where
\begin{equation}\label{eq03}
\overline{\bm{U}}_{k}
=-\Delta\bm{U}_{k},\quad
\bm{D}_{k}
=-\Delta\bm{N}_{k},\quad\forall k\in\mathbb{Z}_{+}.
\end{equation}

\noindent By contrast with the plant (\ref{eq01}), (\ref{eq02}) represents a dynamic system with evolution along the iteration axis.

From the perspective of control systems, (\ref{eq02}) defines a linear discrete control system \cite{am:06}, where we call $\bm{E}_{k}$, $\overline{\bm{U}}_{k}$, and $\bm{D}_{k}$ the $k$-state, $k$-input, and $k$-disturbance for distinction, respectively. This observation motivates us to simply say that a system is $k$-stable if $\lim_{k\to\infty}\bm{E}_{k}=0$. Furthermore, we present the following robust $k$-stability notion when considering the system (\ref{eq02}) under the feedback-based controllers.

\begin{defi}\label{def01}
For the system (\ref{eq02}) integrated with a feedback-based controller, if there exist some class $\mathcal{K}_{\infty}$ functions $\chi_{1}$ and $\chi_{2}$, some class $\mathcal{KL}$ function $\zeta$, and some finite bound $\beta_{0}\geq0$ such that for any bounded initial condition $\bm{E}_{0}$ and uncertainty $\bm{N}_{k}$, the following two properties hold:
\begin{enumerate}
\item
{\it(boundedness):} $\left\|\bm{E}_{k}\right\|\leq\chi_{1}\left(\beta_{\Delta\bm{N}}\right)+\zeta\left(\beta_{0},k\right)$, $\forall k\in\mathbb{Z}_{+}$;

\item
{\it(attractiveness):} $\limsup_{k\to\infty}\left\|\bm{E}_{k}\right\|\leq\chi_{2}\left(\beta_{\Delta\bm{N}}^{ess}\right)$;
\end{enumerate}

\noindent then the resulting closed-loop system is robustly $k$-stable.
\end{defi}

Based on Definition \ref{def01}, we can address the tracking problem of the plant (\ref{eq01}) by designing $k$-inputs for the system (\ref{eq02}) to seek the robust $k$-stability of the tracking error. A direct advantage of this problem transformation is that the $k$-state is available to design feedback controllers for the system (\ref{eq02}). It provides the possibility to handle the tracking problem of the data sequence produced by (\ref{eq01}) in the Kalman state-space framework.

The two properties of Definition \ref{def01} explicitly disclose that for the steady case, the tracking error bound depends continuously upon the variation bound of the iteration-varying uncertainty. It may, however, lead to low tracking accuracies, especially when subjected to iteration-varying uncertainties with fast variations. A further problem of interest is to overcome this drawback for data-driven iterative and learning methods, for which we make an improvement of the $k$-stability in Definition \ref{def01} to propose a $k$-superstability concept in the following definition.

\begin{defi}\label{def02}
For the system (\ref{eq02}) integrated with a feedback-based controller, if there exist some class $\mathcal{K}_{\infty}$ functions $\chi_{1}$ and $\chi_{2}$, some class $\mathcal{KL}$ function $\zeta$, and some finite bound $\beta_{0}\geq0$ such that for any bounded initial condition $\bm{E}_{0}$ and uncertainty $\bm{N}_{k}$, the following two properties hold:
\begin{enumerate}
\item
{\it(boundedness):} $\left\|\bm{E}_{k}\right\|\leq\chi_{1}\left(\beta_{\Delta^{2}\bm{N}}\right)+\zeta\left(\beta_{0},k\right)$, $\forall k\in\mathbb{Z}_{+}$;

\item
{\it(superattractiveness):} $\limsup_{k\to\infty}\left\|\bm{E}_{k}\right\|\leq\chi_{2}\left(\beta_{\Delta^{2}\bm{N}}^{ess}\right)$;
\end{enumerate}

\noindent then the resulting closed-loop system is robustly $k$-superstable.
\end{defi}

\begin{rem}\label{rem010}
For boundedness properties of Definitions \ref{def01} and \ref{def02}, $\beta_{0}$ is a bound generally depending upon the initial condition. The boundedness identifies the overall system performance for all iterations, from which both a unified bound for all iterations and a steady-state bound as the iteration index tends to infinity can be determined. Let us, for example, consider Definition \ref{def01}, and then we can arrive at
\[
\sup_{k\in\mathbb{Z}_{+}}\left\|\bm{E}_{k}\right\|
\leq\zeta\left(\beta_{0},0\right)+\chi_{1}\left(\beta_{\Delta\bm{N}}\right),\quad
\limsup_{k\to\infty}\left\|\bm{E}_{k}\right\|
\leq\chi_{1}\left(\beta_{\Delta\bm{N}}\right)
\]

\noindent in which a peak-to-peak performance evaluation under the zero initial condition is implicitly involved. The similar result works for Definition \ref{def02}. However, a conservative bound is induced for the steady-state case, in contrast to which the bound developed via the (super)attractiveness property is more relaxed and more applicable for reflecting the high-precision tracking tasks.
\end{rem}
\begin{rem}\label{rem02}
Because $\beta_{\Delta^{2}\bm{N}}\leq2\beta_{\Delta\bm{N}}$ and $\beta_{\Delta^{2}\bm{N}}^{ess}\leq2\beta_{\Delta\bm{N}}^{ess}$, we can find that the robustly $k$-superstable systems must be robustly $k$-stable. It reveals that the robust $k$-superstability moves forward the robust $k$-stability. In particular, if $\beta_{\Delta\bm{N}}^{ess}=0$, then we can get $\lim_{k\to\infty}\bm{E}_{k}=0$ from the attractiveness (or superattractiveness). Thus, the $k$-stability can emerge from the robust $k$-stability (or $k$-superstability), provided the variation of the iteration-varying uncertainty disappears. The $k$-stability can also result from the robust $k$-superstability when $\beta_{\Delta^{2}\bm{N}}^{ess}=0$, i.e., the variation of $\bm{N}_{k}$ quasi-disappears (even it may not disappear).
\end{rem}

In what follows, we leverage the notions of robust $k$-stability and $k$-superstability to explore data-driven learning approaches by taking advantage of the rich set of design and analysis tools proposed in the Kalman state-space framework. Our developed results particularly apply to ILC since the plant (\ref{eq01}) is effective in describing a class of controlled systems running repetitively over a fixed time interval (see, e.g., \cite{bta:06,acm:07}).

\section{Nominal Kalman State-Space Design}\label{sec3}

In this section, we explore the nominal design of the Kalman state-space framework for data-driven learning by ignoring the model uncertainty of the plant (\ref{eq01}) (namely, taking $P_{\delta}=0$). We then develop a basic and an ESO-based state-space approaches for data-driven learning, respectively.

\subsection{Basic $k$-State Space}

We reconsider the system (\ref{eq02}) described in the Kalman state-space form, of which we can develop the controllability and ($k$-state feedback) stabilizability properties as follows (for details of controllability and relevant properties of linear systems, we refer the readers to, e.g., \cite[Chapter 3]{am:06}).

\begin{lem}\label{lem1}
For the system (\ref{eq02}), the following are equivalent:
\begin{enumerate}
\item
$(I,P)$ is controllable;

\item
$(I,P)$ is stabilizable;

\item
$P$ is a full-row rank matrix.
\end{enumerate}
\end{lem}

\begin{proof}
``1)$\Rightarrow$2):'' A classical result of linear systems.

``2)$\Rightarrow$3):'' If $(I,P)$ is stabilizable, then there exists some gain matrix $K\in\mathbb{R}^{m\times p}$ such that
\begin{equation}\label{eq04}
\rho\left(I-PK\right)<1.
\end{equation}

\noindent Clearly, (\ref{eq04}) ensures the nonsingularity of $PK$, and hence $P$ has the full-row rank.

``3)$\Rightarrow$1):'' A consequence of the controllability rank criterion thanks to the trivial nonsingularity of the identity matrix.
\end{proof}

As an application of Lemma \ref{lem1}, the following typical design result can be presented by utilizing the controllability property of the system (\ref{eq02}).

\begin{thm}\label{thm1}
Consider the system (\ref{eq02}), and let the $k$-input be given in a $k$-state feedback form of
\begin{equation}\label{eq05}
\overline{\bm{U}}_{k}=-K\bm{E}_{k},\quad\forall k\in\mathbb{Z}_{+}.
\end{equation}

\noindent Then the closed-loop system described in terms of (\ref{eq02}) and (\ref{eq05}) is robustly $k$-stable if and only if the spectral radius condition (\ref{eq04}) holds.
\end{thm}

\begin{proof}
The application of (\ref{eq05}) to (\ref{eq02}) yields
\begin{equation}\label{eq06}
\bm{E}_{k+1}
=\left(I-PK\right)\bm{E}_{k}-\Delta\bm{N}_{k},\quad\forall k\in\mathbb{Z}_{+}.
\end{equation}

\noindent If (\ref{eq04}) holds, then there exists some induced matrix norm that fulfills $\left\|I-PK\right\|<1$ according to \cite[Lemma 5.6.10]{hj:85}. Solving the solution to (\ref{eq06}) leads to
\begin{equation*}\label{}
\bm{E}_{k}
=\left(I-PK\right)^{k}\bm{E}_{0}-\sum_{i=0}^{k-1}\left(I-PK\right)^{k-1-i}\Delta\bm{N}_{i},\quad\forall k\in\mathbb{Z}_{+}
\end{equation*}

\noindent with which we can obtain a boundedness result as
\begin{equation}\label{eq07}
\left\|\bm{E}_{k}\right\|
\leq\left\|I-PK\right\|^{k}\left\|\bm{E}_{0}\right\|+\frac{\Ds\beta_{\Delta\bm{N}}}{\Ds1-\left\|I-PK\right\|},\quad\forall k\in\mathbb{Z}_{+}.
\end{equation}

\noindent In addition, we resort to (\ref{eq06}) and can arrive at an attractiveness result as
\begin{equation}\label{eq08}
\limsup_{k\to\infty}\left\|\bm{E}_{k}\right\|
\leq\frac{\Ds\beta_{\Delta\bm{N}}^{ess}}{\Ds1-\left\|I-PK\right\|}.
\end{equation}

\noindent With (\ref{eq07}) and (\ref{eq08}), we can conclude from Definition \ref{def01} that (\ref{eq06}) is a robustly $k$-stable system.

On the contrary, if the robust $k$-stability of (\ref{eq06}) is given, then it can be validated from the attractiveness of Definition \ref{def01} that (\ref{eq04}) is necessarily required.
\end{proof}

As an equivalent form of (\ref{eq05}), an updating law can be derived for the plant (\ref{eq01}) as
\begin{equation}\label{eq09}
\bm{U}_{k+1}=\bm{U}_{k}+K\bm{E}_{k},\quad\forall k\in\mathbb{Z}_{+}.
\end{equation}

\noindent Then it follows from Theorem \ref{thm1} that the output of the plant (\ref{eq01}) under the updating law (\ref{eq09}) can accomplish the robust tracking of the desired output target. Furthermore, the tracking error can be decreased to a small bound that depends on the variation of the iteration-varying uncertainty, rather than on the uncertainty itself. This robust tracking result is particularly effective in the presence of large iteration-varying uncertainties with relatively small variations.

\subsection{Extended $k$-State Space}

To proceed further with the basic design result for the $k$-state feedback in Theorem \ref{thm1}, we consider involving the information of the iteration-varying uncertainty in the design of the $k$-input $\overline{\bm{U}}_{k}$. We thus define an extended $k$-state as
\[
\overline{\bm{X}}_{k}=\begin{bmatrix}\bm{E}_{k}\\\bm{D}_{k}\end{bmatrix}\in\mathbb{R}^{2p},\quad\forall k\in\mathbb{Z}_{+}.
\]

\noindent Correspondingly, we resort to (\ref{eq02}) and can develop an extended Kalman state-space description as
\begin{equation}\label{eq010}
\left\{\aligned
\overline{\bm{X}}_{k+1}
&=\overline{A}\,\overline{\bm{X}}_{k}+\overline{B}\,\overline{\bm{U}}_{k}+\overline{\bm{D}}_{k}\\
\overline{\bm{Y}}_{k}
&=\overline{C}\,\overline{\bm{X}}_{k}
\endaligned,\quad\forall k\in\mathbb{Z}_{+}\right.
\end{equation}

\noindent of which the output $\overline{\bm{Y}}_{k}\in\mathbb{R}^{p}$, the disturbance $\overline{\bm{D}}_{k}\in\mathbb{R}^{2p}$, and the three system matrices $\overline{A}\in\mathbb{R}^{2p\times2p}$, $\overline{B}\in\mathbb{R}^{2p\times m}$, and $\overline{C}\in\mathbb{R}^{p\times2p}$ are given by
\begin{equation}\label{eq011}
\overline{\bm{Y}}_{k}
=\bm{E}_{k},\quad
\overline{\bm{D}}_{k}
=\begin{bmatrix}0\\\Delta\bm{D}_{k}\end{bmatrix},\quad
\overline{A}
=\begin{bmatrix}
I&I\\
0&I\\
\end{bmatrix},\quad
\overline{B}
=\begin{bmatrix}
P\\
0\\
\end{bmatrix},\quad
\overline{C}
=\begin{bmatrix}
I&0\\
\end{bmatrix}.
\end{equation}

\noindent Unlike the $k$-state, the extended $k$-state of the system (\ref{eq010}) is no longer available due to the presence of the unknown iteration-varying uncertainty. Let $F\in\mathbb{R}^{p\times2p}$ be given by
\[
F=\begin{bmatrix}
0&I\\
\end{bmatrix}
\]

\noindent and then with (\ref{eq03}), $\overline{\bm{D}}_{k}$ can be written as
\begin{equation}\label{eq012}
\overline{\bm{D}}_{k}
=F^{\tp}\Delta\bm{D}_{k}
=-F^{\tp}\Delta^{2}\bm{N}_{k},\quad\forall k\in\mathbb{Z}_{+}.
\end{equation}

For the system (\ref{eq010}), we can get the following controllability and observability properties.

\begin{lem}\label{lem2}
For the system (\ref{eq010}), the following hold:
\begin{enumerate}
\item
$\left(\overline{A},\overline{B},\overline{C}\right)$ is observable but not controllable;

\item
$\left(\overline{A},\overline{B},\overline{C}\right)$ is a standard controllable decomposition if and only if $P$ is a full-row rank matrix.
\end{enumerate}
\end{lem}

\begin{proof}
1): From (\ref{eq011}), we can verify that $\left(\overline{A},\overline{C}\right)$ is observable, but $\left(\overline{A},\overline{B}\right)$ is not controllable.

2): $\left(\overline{A},\overline{B},\overline{C}\right)$ is a standard controllable decomposition if and only if $\left(I,P\right)$ is controllable. Then by leveraging Lemma \ref{lem1}, we can develop the second statement in this lemma.
\end{proof}

Even though the extended $k$-state $\overline{\bm{X}}_{k}$ is no longer available, Lemma \ref{lem2} provides a basic guarantee that $\overline{\bm{X}}_{k}$ can be estimated through an observer-based design. We thus denote an observer $k$-state of $\overline{\bm{X}}_{k}$ as
\[
\widehat{\overline{\bm{X}}}_{k}=\begin{bmatrix}\widehat{\bm{E}}_{k}\\\widehat{\bm{D}}_{k}\end{bmatrix}\in\mathbb{R}^{2p},\quad\forall k\in\mathbb{Z}_{+}
\]

\noindent where $\widehat{\bm{E}}_{k}$ and $\widehat{\bm{D}}_{k}$ in fact denote the estimations of $\bm{E}_{k}$ and $\bm{D}_{k}$, respectively. Then for the system (\ref{eq010}), we can design an ESO in the form of
\begin{equation}\label{eq013}
\left\{\aligned
\widehat{\overline{\bm{X}}}_{k+1}
&=\overline{A}\,\widehat{\overline{\bm{X}}}_{k}
+\overline{B}\,\overline{\bm{U}}_{k}+\overline{L}\left(\overline{\bm{Y}}_{k}-\widehat{\overline{\bm{Y}}}_{k}\right)\\
\widehat{\overline{\bm{Y}}}_{k}
&=\overline{C}\,\widehat{\overline{\bm{X}}}_{k}
\endaligned,\quad\forall k\in\mathbb{Z}_{+}\right.
\end{equation}

\noindent in which $\widehat{\overline{\bm{Y}}}_{k}\in\mathbb{R}^{p}$ is the output estimation of $\overline{\bm{Y}}_{k}$, and $\overline{L}\in\mathbb{R}^{2p\times p}$ is the observer gain matrix given by
\begin{equation*}
\overline{L}=
\begin{bmatrix}
L_{1}\\
L_{2}\\
\end{bmatrix}~\hbox{with}~L_{i}\in\mathbb{R}^{p\times p},\quad i=1,2.
\end{equation*}

\noindent Let the observation error of the extended $k$-state be defined by $\widetilde{\overline{\bm{X}}}_{k}=\overline{\bm{X}}_{k}-\widehat{\overline{\bm{X}}}_{k}$, or more precisely,
\[\widetilde{\overline{\bm{X}}}_{k}
=\begin{bmatrix}\widetilde{\bm{E}}_{k}\\\widetilde{\bm{D}}_{k}\end{bmatrix}
=\begin{bmatrix}\bm{E}_{k}-\widehat{\bm{E}}_{k}\\\bm{D}_{k}-\widehat{\bm{D}}_{k}\end{bmatrix},\quad\forall k\in\mathbb{Z}_{+}.
\]

For the ESO (\ref{eq013}) as well as the system (\ref{eq010}), we can develop the following properties of them, especially through leveraging the observability result of Lemma \ref{lem2}.

\begin{lem}\label{lem3}
In a compact Kalman state-space description, the ESO (\ref{eq013}) can be written as
\begin{equation}\label{eq014}
\aligned
\widehat{\overline{\bm{X}}}_{k+1}
&=\left(\overline{A}-\overline{L}\,\overline{C}\right)\widehat{\overline{\bm{X}}}_{k}
+\overline{B}\,\overline{\bm{U}}_{k}+\overline{L}\,\overline{\bm{Y}}_{k}\\
&=\left(\overline{A}-\overline{L}\,\overline{C}\right)\widehat{\overline{\bm{X}}}_{k}
+\overline{B}\,\overline{\bm{U}}_{k}+\overline{L}\bm{E}_{k},\quad\forall k\in\mathbb{Z}_{+}
\endaligned
\end{equation}

\noindent and the extended $k$-state observation error can be described by
\begin{equation}\label{eq015}
\widetilde{\overline{\bm{X}}}_{k+1}
=\left(\overline{A}-\overline{L}\,\overline{C}\right)\widetilde{\overline{\bm{X}}}_{k}
+\overline{\bm{D}}_{k},\quad\forall k\in\mathbb{Z}_{+}.
\end{equation}

\noindent Further, there exists some ESO (\ref{eq014}) such that the error system (\ref{eq015}) has the following properties:
\begin{enumerate}
\item
$\widetilde{\overline{\bm{X}}}_{k}$ is bounded, i.e., for some class $\mathcal{K}_{\infty}$ function $\chi_{1}$, some class $\mathcal{KL}$ function $\zeta$, and some finite bound $\beta_{0}\geq0$,
\begin{equation}\label{eq092}
\left\|\widetilde{\overline{\bm{X}}}_{k}\right\|
\leq\chi_{1}\left(\beta_{\Delta^{2}\bm{N}}\right)+\zeta\left(\beta_{0},k\right),\quad\forall k\in\mathbb{Z}_{+}
\end{equation}

\item
$\widetilde{\overline{\bm{X}}}_{k}$ has a superattractiveness property for some class $\mathcal{K}_{\infty}$ function $\chi_{2}$ as
\begin{equation}\label{eq016}
\limsup_{k\to\infty}\left\|\widetilde{\overline{\bm{X}}}_{k}\right\|
\leq\chi_{2}\left(\beta_{\Delta^{2}\bm{N}}^{ess}\right)
\end{equation}
\end{enumerate}

\noindent if and only if the spectral radius of $\overline{A}-\overline{L}\,\overline{C}$ satisfies
\begin{equation}\label{eq017}
\rho\left(\overline{A}-\overline{L}\,\overline{C}\right)<1.
\end{equation}

\noindent In addition, there always exists some gain matrix $\overline{L}$ to achieve the spectral radius condition (\ref{eq017}).
\end{lem}

\begin{proof}
We can easily validate (\ref{eq014}) and (\ref{eq015}) with (\ref{eq010}) and (\ref{eq013}). In Lemma \ref{lem2}, the observability of the matrix pair $\left(\overline{A},\overline{C}\right)$ is validated, which guarantees the existence of some gain matrix $\overline{L}$ to fulfill (\ref{eq017}). This, together with \cite[Lemma 5.6.10]{hj:85}, leads to $\left\|\overline{A}-\overline{L}\,\overline{C}\right\|<1$ in the sense of some induced matrix norm. In the same way as used in the proof of Theorem \ref{thm1}, we combine (\ref{eq012}) and (\ref{eq015}) to deduce
\[
\widetilde{\overline{\bm{X}}}_{k}
=\left(\overline{A}-\overline{L}\,\overline{C}\right)^{k}\widetilde{\overline{\bm{X}}}_{0}
-\sum_{i=0}^{k-1}\left(\overline{A}-\overline{L}\,\overline{C}\right)^{k-1-i}F^{\tp}\Delta^{2}\bm{N}_{i},\quad\forall k\in\mathbb{Z}_{+}
\]

\noindent which leads to
\[
\left\|\widetilde{\overline{\bm{X}}}_{k}\right\|
\leq\left\|\overline{A}-\overline{L}\,\overline{C}\right\|^{k}\left\|\widetilde{\overline{\bm{X}}}_{0}\right\|
+\frac{\Ds\beta_{\Delta^{2}\bm{N}}}{\Ds1-\left\|\overline{A}-\overline{L}\,\overline{C}\right\|},\quad\forall k\in\mathbb{Z}_{+}
\]

\noindent and
\[\limsup_{k\to\infty}\left\|\widetilde{\overline{\bm{X}}}_{k}\right\|
\leq\frac{\Ds\beta_{\Delta^{2}\bm{N}}^{ess}}{\Ds1-\left\|\overline{A}-\overline{L}\,\overline{C}\right\|}.
\]

\noindent With these two results, this lemma can be obtained.
\end{proof}

From Lemma \ref{lem3}, it follows that the observation error $\widetilde{\overline{\bm{X}}}_{k}$ can always be bounded with a superattractiveness property related to the variation rate $\Delta^{2}\bm{N}_{k}$ of the iteration-varying uncertainty $\bm{N}_{k}$. Particularly, (\ref{eq092}) guarantees $\left\|\widetilde{\overline{\bm{X}}}_{k}\right\|\leq\beta_{\widetilde{\overline{\bm{X}}}}$, $\forall k\in\mathbb{Z}_{+}$ for some unified bound $\beta_{\widetilde{\overline{\bm{X}}}}\geq0$, while (\ref{eq016}) represents a zero-convergent observation error, namely, $\lim_{k\to\infty}\widetilde{\overline{\bm{X}}}_{k}=0$ when the variation of the iteration-varying uncertainty quasi-disappears. This reveals that we may reconstruct the iteration-varying uncertainty based on the observer $k$-state of our designed ESO (\ref{eq014}).

With Lemmas \ref{lem2} and \ref{lem3}, we proceed to design an ESO-based feedback controller. Let a gain matrix be denoted as
\[
\overline{K}\triangleq\begin{bmatrix}K&H\end{bmatrix}\in\mathbb{R}^{m\times2p}
\]

\noindent for some $K\in\mathbb{R}^{m\times p}$ (the same as (\ref{eq05})) and $H\in\mathbb{R}^{m\times p}$, and then the following ESO-based feedback result can be developed for the system (\ref{eq02}).

\begin{lem}\label{lem4}
For the system (\ref{eq02}) with the ESO (\ref{eq014}), if an ESO-based feedback controller is applied as
\begin{equation}\label{eq018}
\overline{\bm{U}}_{k}
=-\overline{K}\,\widehat{\overline{\bm{X}}}_{k}
=-K\widehat{\bm{E}}_{k}-H\widehat{\bm{D}}_{k},\quad\forall k\in\mathbb{Z}_{+}
\end{equation}

\noindent then the (eigenvalue) separation principle is satisfied such that the feedback controller (\ref{eq018}) and the ESO (\ref{eq014}) can be designed separately. Particularly, the feedback gain matrix $K$ in (\ref{eq018}) can be synthesized in the same way with that used in (\ref{eq05}).
\end{lem}

\begin{proof}
From (\ref{eq02}), (\ref{eq014}), and (\ref{eq018}), we can achieve a closed-loop state-space description as
\begin{equation}\label{eq019}
\begin{bmatrix}
\bm{E}_{k+1}\\
\widehat{\overline{\bm{X}}}_{k+1}
\end{bmatrix}
=\begin{bmatrix}
I&-P\overline{K}\\
\overline{L}&\overline{A}-\overline{L}\,\overline{C}-\overline{B}\,\overline{K}
\end{bmatrix}
\begin{bmatrix}
\bm{E}_{k}\\
\widehat{\overline{\bm{X}}}_{k}
\end{bmatrix}
+\begin{bmatrix}I\\0\end{bmatrix}\bm{D}_{k},\quad\forall k\in\mathbb{Z}_{+}.
\end{equation}

\noindent By resorting to a nonsingular linear transformation matrix as
\[
\begin{bmatrix}
I&0\\
-\overline{C}^{\tp}&I
\end{bmatrix}
\]

\noindent and integrating $\overline{C}^{\tp}P=\overline{B}$, $\overline{A}\,\overline{C}^{\tp}=\overline{C}^{\tp}$, $\overline{C}\,\overline{C}^{\tp}=I$, and $\overline{K}\,\overline{C}^{\tp}=K$, we can deduce for the system (\ref{eq019}) that
\begin{equation}\label{eq020}
\aligned
\begin{bmatrix}
I&0\\
-\overline{C}^{\tp}&I
\end{bmatrix}
&\begin{bmatrix}
I&-P\overline{K}\\
\overline{L}&\overline{A}-\overline{L}\,\overline{C}-\overline{B}\,\overline{K}
\end{bmatrix}
\begin{bmatrix}
I&0\\
-\overline{C}^{\tp}&I
\end{bmatrix}^{-1}
=\begin{bmatrix}
I-PK
&-P\overline{K}\\
0
&\overline{A}-\overline{L}\,\overline{C}
\end{bmatrix}.
\endaligned
\end{equation}

\noindent With the upper block-triangular form of (\ref{eq020}), we can establish the separation principle result of this lemma.
\end{proof}

Based on Lemmas \ref{lem2}, \ref{lem3} and \ref{lem4}, we now propose the following theorem to establish an ESO-based design result.

\begin{thm}\label{thm2}
Consider the system (\ref{eq02}) with the ESO (\ref{eq014}), and let the $k$-input be applied in the ESO-based feedback form of (\ref{eq018}). Then the closed-loop system given by (\ref{eq02}), (\ref{eq014}), and (\ref{eq018}) is robustly $k$-stable if and only if the spectral radius conditions (\ref{eq04}) and (\ref{eq017}) both hold. Further, when selecting the gain matrix $H$ as
\begin{equation}\label{eq021}
H=P^{\tp}\left(PP^{\tp}\right)^{-1}
\end{equation}

\noindent the robust $k$-superstability can be achieved if and only if both (\ref{eq04}) and (\ref{eq017}) are satisfied.
\end{thm}

\begin{proof}
We resort to (\ref{eq03}) and (\ref{eq019}) and can deduce
\begin{equation}\label{eq022}
\aligned
\bm{E}_{k+1}
&=\bm{E}_{k}
-P\overline{K}\,\widehat{\overline{\bm{X}}}_{k}
+\bm{D}_{k}\\
&=\left(I-PK\right)\bm{E}_{k}
+P\overline{K}\,\widetilde{\overline{\bm{X}}}_{k}
-\left(I-PH\right)\Delta\bm{N}_{k},\quad\forall k\in\mathbb{Z}_{+}.
\endaligned
\end{equation}

\noindent By combining (\ref{eq012}) and (\ref{eq015}) with (\ref{eq022}), we can derive
\begin{equation}\label{eq023}
\aligned
\begin{bmatrix}
\bm{E}_{k+1}\\
\widetilde{\overline{\bm{X}}}_{k+1}
\end{bmatrix}
&=\begin{bmatrix}I-PK&P\overline{K}\\
0&\overline{A}-\overline{L}\,\overline{C}\end{bmatrix}
\begin{bmatrix}
\bm{E}_{k}\\
\widetilde{\overline{\bm{X}}}_{k}
\end{bmatrix}
-\begin{bmatrix}I-PH&0\\0&F^{\tp}\end{bmatrix}
\begin{bmatrix}\Delta\bm{N}_{k}\\\Delta^{2}\bm{N}_{k}\end{bmatrix},\quad\forall k\in\mathbb{Z}_{+}.
\endaligned
\end{equation}

\noindent If (\ref{eq04}) and (\ref{eq017}) hold, then for (\ref{eq023}), we have
\begin{equation}\label{eq024}
\rho\left(\begin{bmatrix}I-PK&P\overline{K}\\
0&\overline{A}-\overline{L}\,\overline{C}\end{bmatrix}\right)<1
\end{equation}

\noindent which, together with \cite[Lemma 5.6.10]{hj:85}, implies the existence of some induced matrix norm such that
\begin{equation}\label{eq025}
\left\|\begin{bmatrix}I-PK&P\overline{K}\\
0&\overline{A}-\overline{L}\,\overline{C}\end{bmatrix}\right\|<1.
\end{equation}

\noindent Hence, we consider the system (\ref{eq023}) under the condition (\ref{eq025}), and can employ $\bm{E}_{k}=\left[I~0\right]\left[\bm{E}^{\tp}_{k}~\widetilde{\overline{\bm{X}}}^{\tp}_{k}\right]^{\tp}$ to obtain
\begin{equation*}\label{}
\aligned
\left\|\bm{E}_{k}\right\|
&\leq\left\|\begin{bmatrix}I-PK&P\overline{K}\\0&\overline{A}-\overline{L}\,\overline{C}\end{bmatrix}\right\|^{k}
\left\|\begin{bmatrix}\bm{E}_{0}\\\widetilde{\overline{\bm{X}}}_{0}\end{bmatrix}\right\|
+\frac{\Ds3\beta_{\Delta\bm{N}}\left\|\begin{bmatrix}I-PH&0\\0&F^{\tp}\end{bmatrix}\right\|}
{\Ds1-\left\|\begin{bmatrix}I-PK&P\overline{K}\\0&\overline{A}-\overline{L}\,\overline{C}\end{bmatrix}\right\|}
\endaligned,\quad\forall k\in\mathbb{Z}_{+}
\end{equation*}

\noindent and
\begin{equation*}\label{}
\limsup_{k\to\infty}\left\|\bm{E}_{k}\right\|
\leq\frac{\Ds3\beta_{\Delta\bm{N}}^{ess}\left\|\begin{bmatrix}I-PH&0\\0&F^{\tp}\end{bmatrix}\right\|}
{\Ds1-\left\|\begin{bmatrix}I-PK&P\overline{K}\\0&\overline{A}-\overline{L}\,\overline{C}\end{bmatrix}\right\|}.
\end{equation*}

\noindent Then it follows immediately from Definition \ref{def01} that the robust $k$-stability is realized.

To proceed, if we further consider (\ref{eq021}) that can be ensured under the condition (\ref{eq04}), then (\ref{eq023}) collapses into
\begin{equation}\label{eq026}
\begin{bmatrix}
\bm{E}_{k+1}\\
\widetilde{\overline{\bm{X}}}_{k+1}
\end{bmatrix}
=\begin{bmatrix}I-PK&P\overline{K}\\
0&\overline{A}-\overline{L}\,\overline{C}\end{bmatrix}
\begin{bmatrix}
\bm{E}_{k}\\
\widetilde{\overline{\bm{X}}}_{k}
\end{bmatrix}
-\begin{bmatrix}0\\F^{\tp}\end{bmatrix}
\Delta^{2}\bm{N}_{k},\quad\forall k\in\mathbb{Z}_{+}.
\end{equation}

\noindent For the system (\ref{eq026}) under the condition (\ref{eq025}), we can validate
\begin{equation*}\label{}
\aligned
\left\|\bm{E}_{k}\right\|
&\leq\left\|\begin{bmatrix}I-PK&P\overline{K}\\0&\overline{A}-\overline{L}\,\overline{C}\end{bmatrix}\right\|^{k}
\left\|\begin{bmatrix}\bm{E}_{0}\\\widetilde{\overline{\bm{X}}}_{0}\end{bmatrix}\right\|
+\frac{\Ds\beta_{\Delta^{2}\bm{N}}}{\Ds1-\left\|\begin{bmatrix}I-PK&P\overline{K}\\0&\overline{A}-\overline{L}\,\overline{C}\end{bmatrix}\right\|}
\endaligned,\quad\forall k\in\mathbb{Z}_{+}
\end{equation*}

\noindent and
\begin{equation*}\label{}
\limsup_{k\to\infty}\left\|\bm{E}_{k}\right\|
\leq\frac{\Ds\beta_{\Delta^{2}\bm{N}}^{ess}}{\Ds1-\left\|\begin{bmatrix}I-PK&P\overline{K}\\0&\overline{A}-\overline{L}\,\overline{C}\end{bmatrix}\right\|}.
\end{equation*}

\noindent Consequently, the robust $k$-superstability is achieved according to Definition \ref{def02}.

For the necessity, we can validate from (\ref{eq023}) and (\ref{eq026}) that the robust $k$-stability and $k$-superstability both need the satisfaction of (\ref{eq024}). It is obvious that (\ref{eq024}) holds if and only if (\ref{eq04}) and (\ref{eq017}) both hold.
\end{proof}

\begin{rem}\label{rem03}
In contrast to Theorem \ref{thm1}, Theorem \ref{thm2} reveals that the ESO-based feedback controller can maintain the properties of the basic $k$-state feedback controller in achieving the robust $k$-stability, regardless of iteration-varying uncertainties. It also provides a feasible way to improve the robust $k$-stability results by leveraging the ESO-based design and analysis to realize the robust $k$-superstability. In addition, it is worth emphasizing that two stability conditions (\ref{eq04}) and (\ref{eq017}) of Theorem \ref{thm2} are separate from each other and associated with the gain matrices $K$ and $\overline{L}$, respectively. Though the other gain matrix $H$ is independent of these stability conditions, the selection of $H$ plays an important role in influencing the stability performances.
\end{rem}

From the ESO-based feedback controller (\ref{eq018}), an equivalent updating law can be induced as
\begin{equation}\label{eq027}
\bm{U}_{k+1}=\bm{U}_{k}+K\widehat{\bm{E}}_{k}+H\widehat{\bm{D}}_{k},\quad\forall k\in\mathbb{Z}_{+}.
\end{equation}

\noindent It can be seen from Theorem \ref{thm2} that for the plant (\ref{eq01}) under the ESO-based updating law (\ref{eq027}), not only can the robust tracking objective be realized in spite of iteration-varying uncertainties, but also the tracking error can be ensured to decrease to a small bound. In contrast to the use of (\ref{eq09}), the ESO-based information helps (\ref{eq027}) to overcome the drawback caused by the continuous dependence of the tracking error bound on the variation bound of the iteration-varying uncertainty. It is worth noticing that the perfect tracking of the desired output target can be realized by the output of the plant (\ref{eq01}) under the ESO-based updating law (\ref{eq027}) although the variation of the iteration-varying uncertainty is ensured to not disappear, but only quasi-disappear.

\subsection{Improvement of ESO-Based Design}

The basic design result of Theorem \ref{thm1} takes full advantage of the available $k$-state information that, however, is not leveraged by the ESO-based design result of Theorem \ref{thm2}. Of specific use in the feedback design is only the ESO information in Theorem \ref{thm2}. Next, we benefit from that the information of $\bm{E}_{k}$ is available and try to propose an alternative design method for ESO-based data-driven learning. Towards this end, we show the following lemma to provide a fundamental guarantee with the separation principle between the designs of ESO and feedback controller of the system (\ref{eq02}).

\begin{lem}\label{lem5}
For the system (\ref{eq02}), let the ESO (\ref{eq014}) and an ESO-based feedback controller be applied as
\begin{equation}\label{eq028}
\overline{\bm{U}}_{k}
=-K\bm{E}_{k}-H\widehat{\bm{D}}_{k},\quad\forall k\in\mathbb{Z}_{+}.
\end{equation}

\noindent Then the same result of the (eigenvalue) separation principle as Lemma \ref{lem4} also works between the feedback controller (\ref{eq028}) and the ESO (\ref{eq014}), and further the syntheses of three gain matrices $K$, $H$, and $\overline{L}$ are separate from each other.
\end{lem}

\begin{proof}
Due to $\widehat{\bm{D}}_{k}=F\widehat{\overline{\bm{X}}}_{k}$ and by integrating (\ref{eq02}), (\ref{eq014}), and (\ref{eq028}), we can obtain a closed-loop state-space description as
\begin{equation}\label{eq029}
\begin{bmatrix}
\bm{E}_{k+1}\\
\widehat{\overline{\bm{X}}}_{k+1}
\end{bmatrix}
=\begin{bmatrix}
I-PK&-PHF\\
\overline{L}-\overline{B}K&\overline{A}-\overline{L}\,\overline{C}-\overline{B}HF
\end{bmatrix}
\begin{bmatrix}
\bm{E}_{k}\\
\widehat{\overline{\bm{X}}}_{k}
\end{bmatrix}
+\begin{bmatrix}I\\0\end{bmatrix}\bm{D}_{k},\quad\forall k\in\mathbb{Z}_{+}.
\end{equation}

\noindent With $\overline{C}^{\tp}P=\overline{B}$, $\overline{A}\,\overline{C}^{\tp}=\overline{C}^{\tp}$, $\overline{C}\,\overline{C}^{\tp}=I$, and $F\overline{C}^{\tp}=0$, we perform the same nonsingular linear transformation as implemented in (\ref{eq020}) for the system (\ref{eq029}), 
%
and then we can validate
\begin{equation}\label{eq030}
\aligned
\begin{bmatrix}
I&0\\
-\overline{C}^{\tp}&I
\end{bmatrix}
&\begin{bmatrix}
I-PK&-PHF\\
\overline{L}-\overline{B}K&\overline{A}-\overline{L}\,\overline{C}-\overline{B}HF
\end{bmatrix}
\begin{bmatrix}
I&0\\
-\overline{C}^{\tp}&I
\end{bmatrix}^{-1}
=\begin{bmatrix}
I-PK
&-PHF\\
0
&\overline{A}-\overline{L}\,\overline{C}
\end{bmatrix}.
\endaligned
\end{equation}

\noindent We can obviously accomplish the separation principle between the syntheses of (\ref{eq014}) and (\ref{eq028}) from the upper block-triangular form of (\ref{eq030}). In particular, (\ref{eq030}) implies that $K$, $H$, and $\overline{L}$ can be synthesized separately from each other.
\end{proof}

Motivated by the separation result of Lemma \ref{lem5}, we propose the following ESO-based design result by applying (\ref{eq028}) to the system (\ref{eq02}).

\begin{thm}\label{thm3}
Consider the system (\ref{eq02}) with the ESO (\ref{eq014}), and let the $k$-input be applied in the ESO-based feedback form of (\ref{eq028}). Then the closed-loop system given by (\ref{eq02}), (\ref{eq014}), and (\ref{eq028}) is robustly $k$-stable if and only if the spectral radius conditions (\ref{eq04}) and (\ref{eq017}) both hold. Moreover, when adopting the selection candidate (\ref{eq021}) of $H$, the robust $k$-superstability can be achieved if and only if both (\ref{eq04}) and (\ref{eq017}) are satisfied. In particular, there exists some class $\mathcal{K}_{\infty}$ function $\chi$ such that
\begin{equation}\label{eq031}
\limsup_{k\to\infty}\left\|\bm{E}_{k}\right\|
\leq\chi\left(\limsup_{k\to\infty}\left\|\widetilde{\bm{D}}_{k}\right\|\right).
\end{equation}
\end{thm}

\begin{proof}
The necessity of this theorem is the same as that of Theorem \ref{thm2}. Next, we show the proof of the sufficiency results of this theorem. We further explore (\ref{eq029}) to obtain
\begin{equation}\label{eq032}
\aligned
\bm{E}_{k+1}
&=\left(I-PK\right)\bm{E}_{k}
-PHF\widehat{\overline{\bm{X}}}_{k}+\bm{D}_{k}\\
&=\left(I-PK\right)\bm{E}_{k}
+PHF\widetilde{\overline{\bm{X}}}_{k}
+\left(I-PH\right)\bm{D}_{k}\\
&=\left(I-PK\right)\bm{E}_{k}
+PHF\widetilde{\overline{\bm{X}}}_{k}
-\left(I-PH\right)\Delta\bm{N}_{k},\quad\forall k\in\mathbb{Z}_{+}
\endaligned
\end{equation}

\noindent which, together with (\ref{eq012}) and (\ref{eq015}), yields
\begin{equation}\label{eq033}
\aligned
\begin{bmatrix}
\bm{E}_{k+1}\\
\widetilde{\overline{\bm{X}}}_{k+1}
\end{bmatrix}
&=\begin{bmatrix}I-PK&PHF\\
0&\overline{A}-\overline{L}\,\overline{C}\end{bmatrix}
\begin{bmatrix}
\bm{E}_{k}\\
\widetilde{\overline{\bm{X}}}_{k}
\end{bmatrix}
-\begin{bmatrix}I-PH&0\\0&F^{\tp}\end{bmatrix}
\begin{bmatrix}\Delta\bm{N}_{k}\\\Delta^{2}\bm{N}_{k}\end{bmatrix},\quad\forall k\in\mathbb{Z}_{+}.
\endaligned
\end{equation}

\noindent For (\ref{eq033}) under both conditions (\ref{eq04}) and (\ref{eq017}), we can accomplish the robust $k$-stability by following the same steps as the proof of the robust $k$-stability in Theorem \ref{thm2}.

When using the selection of $H$ in (\ref{eq021}) that can be guaranteed by the condition (\ref{eq04}), we have $PH=I$, and thus (\ref{eq033}) becomes
\begin{equation}\label{eq034}
\begin{bmatrix}
\bm{E}_{k+1}\\
\widetilde{\overline{\bm{X}}}_{k+1}
\end{bmatrix}
=\begin{bmatrix}I-PK&F\\
0&\overline{A}-\overline{L}\,\overline{C}\end{bmatrix}
\begin{bmatrix}
\bm{E}_{k}\\
\widetilde{\overline{\bm{X}}}_{k}
\end{bmatrix}
-\begin{bmatrix}0\\F^{\tp}\end{bmatrix}
\Delta^{2}\bm{N}_{k},\quad\forall k\in\mathbb{Z}_{+}.
\end{equation}

\noindent Similarly to the proof of Theorem \ref{thm2}, the robust $k$-superstability can be established for (\ref{eq034}) when both conditions (\ref{eq04}) and (\ref{eq017}) are satisfied. In particular, the use of (\ref{eq021}) to (\ref{eq032}) yields
\begin{equation}\label{eq035}
\bm{E}_{k+1}
=\left(I-PK\right)\bm{E}_{k}
+\widetilde{\bm{D}}_{k},\quad\forall k\in\mathbb{Z}_{+}.
\end{equation}

\noindent Because we can deduce $\left\|I-PK\right\|<1$ for some induced matrix norm under the condition (\ref{eq04}), we exploit (\ref{eq035}) to arrive at
\[
\limsup_{k\to\infty}\left\|\bm{E}_{k}\right\|
\leq\frac{\Ds\limsup_{k\to\infty}\left\|\widetilde{\bm{D}}_{k}\right\|}{\Ds1-\left\|I-PK\right\|}
\]

\noindent from which (\ref{eq031}) is immediate.
\end{proof}

\begin{rem}\label{rem04}
In contrast with (\ref{eq05}), (\ref{eq028}) makes an improvement to the $k$-state feedback by leveraging the ESO-based estimation information of the iteration-varying uncertainty. This provides a possible way to strengthen the robust $k$-stability to the robust $k$-superstability, as disclosed in Theorem \ref{thm3}. Further, it follows clearly from (\ref{eq031}) that the tracking error depends continuously on the observation error regarding the variation of the iteration-varying uncertainty and vanishes especially when this variation quasi-disappears. The comparison of the ESO-based feedbacks (\ref{eq018}) and (\ref{eq028}) implies that (\ref{eq028}) not only can maintain the same stability results as (\ref{eq018}), but leads to an additional attractiveness result (\ref{eq031}). Note that with the use of (\ref{eq018}), we can employ (\ref{eq021}) and (\ref{eq026}) to only describe the tracking error as
\begin{equation}\label{eq036}
\bm{E}_{k+1}
=\left(I-PK\right)\bm{E}_{k}
+\widetilde{\bm{D}}_{k}
+PK\widetilde{\bm{E}}_{k},\quad\forall k\in\mathbb{Z}_{+}.
\end{equation}

\noindent By comparing (\ref{eq036}) with (\ref{eq035}), we can easily see that the use of (\ref{eq028}) rather than (\ref{eq018}) removes the influence of the observation error $\widetilde{\bm{E}}_{k}$ on the tracking error.
\end{rem}
\begin{rem}\label{rem05}
It is worth highlighting that for Theorems \ref{thm2} and \ref{thm3}, the syntheses of three gain matrices $\overline{L}$, $K$, and $H$ are separate from each other. The synthesis of the ESO (\ref{eq014}) resorts only to determining $\overline{L}$, for which a necessary and sufficient guarantee is provided by (\ref{eq017}). This implementation is independent of the plant (\ref{eq01}) based on (\ref{eq011}). By contrast, the synthesis of feedback controllers is tied closely to the plant (\ref{eq01}), of which (\ref{eq04}) provides a necessary and sufficient condition for the selection of $K$, and (\ref{eq021}) is a desired selection candidate of $H$.
\end{rem}

For the robust tracking problem of the plant (\ref{eq01}), an updating law can be equivalently derived from the ESO-based feedback controller (\ref{eq028}) as
\begin{equation}\label{eq037}
\bm{U}_{k+1}=\bm{U}_{k}+K\bm{E}_{k}+H\widehat{\bm{D}}_{k},\quad\forall k\in\mathbb{Z}_{+}.
\end{equation}

\noindent It can be seen that the $k$-state $\bm{E}_{k}$ is leveraged in (\ref{eq037}), instead of employing its estimation $\widehat{\bm{E}}_{k}$ as in (\ref{eq027}). This not only maintains the same robust tracking results that are described by the robust $k$-stability and $k$-superstability, but also further strengthens the tracking performance of the plant (\ref{eq01}) by rendering the tracking error dependent only on the observation error for the variation of the iteration-varying uncertainty, as is reflected by (\ref{eq031}). In addition, (\ref{eq037}) betters the tracking performance of (\ref{eq09}) by taking advantage of the ESO-based estimation information about the variation of the iteration-varying uncertainty.

\section{Robust Design Under Model Uncertainties}\label{sec4}

In this section, we establish the robust design of the Kalman state-space framework for data-driven learning such that it can simultaneously accommodate the unknown model uncertainty of the plant (\ref{eq01}) when addressing the unknown iteration-varying external uncertainty. Since $P_{\delta}\neq0$ emerges, challenging issues naturally arise. In particular, the nominal design results derived in the previous Section \ref{sec3} may no longer work because of their heavy dependence upon the model information of the plant (\ref{eq01}). Take for example the ESO (\ref{eq014}) that uses the exact information of $P$, and so does the ESO-based feedback controller design of (\ref{eq018}) and (\ref{eq028}) especially when adopting the selection candidate of $H$ in (\ref{eq021}). To deal with these issues resulting from the model uncertainty of (\ref{eq01}), we without loss of generality aim at how to design the ESO-based feedback controller (\ref{eq028}) in the presence of iteration-varying uncertainties.

\subsection{Robustness Analysis Against Model Uncertainties}

We first remove the effect of the model uncertainty $P_{\delta}$ on the ESO (\ref{eq014}). From (\ref{eq011}), we correspondingly denote $\overline{B}=\overline{B}_{0}+\overline{B}_{\delta}$, where $\overline{B}_{0}$ and $\overline{B}_{\delta}$ fulfill
\[
\overline{B}_{0}=
\begin{bmatrix}
P_{0}\\
0\\
\end{bmatrix},\quad
\overline{B}_{\delta}=
\begin{bmatrix}
P_{\delta}\\
0\\
\end{bmatrix}.
\]

\noindent Then a nominal ESO is presented with the following Kalman state-space description instead of (\ref{eq014}):
\begin{equation}\label{eq038}
\widehat{\overline{\bm{X}}}_{k+1}
=\left(\overline{A}-\overline{L}\,\overline{C}\right)\widehat{\overline{\bm{X}}}_{k}
+\overline{B}_{0}\overline{\bm{U}}_{k}+\overline{L}\bm{E}_{k},\quad\forall k\in\mathbb{Z}_{+}
\end{equation}

\noindent which is realizable with the use of known matrices or matrices to be determined. Based on the ESO (\ref{eq038}), the application of the ESO-based feedback controller (\ref{eq028}) to the system (\ref{eq02}) produces a closed-loop state-space description as
\begin{equation}\label{eq039}
\aligned
\begin{bmatrix}
\bm{E}_{k+1}\\
\widehat{\overline{\bm{X}}}_{k+1}
\end{bmatrix}
&=\begin{bmatrix}
I-PK&-PHF\\
\overline{L}-\overline{B}_{0}K&\overline{A}-\overline{L}\,\overline{C}-\overline{B}_{0}HF
\end{bmatrix}
\begin{bmatrix}
\bm{E}_{k}\\
\widehat{\overline{\bm{X}}}_{k}
\end{bmatrix}
-\begin{bmatrix}I\\0\end{bmatrix}\Delta\bm{N}_{k},\quad\forall k\in\mathbb{Z}_{+}.
\endaligned
\end{equation}

By inserting $\widehat{\overline{\bm{X}}}_{k}=\left[\widehat{\bm{E}}^{\tp}_{k}~\widehat{\bm{D}}^{\tp}_{k}\right]^{\tp}$, we can arrive at a nonsingular linear transformation of (\ref{eq039}) as follows:
\[
\begin{bmatrix}
\bm{E}_{k}\\
\begin{bmatrix}-\widetilde{\bm{E}}_{k}\\
\widehat{\bm{D}}_{k}\end{bmatrix}
\end{bmatrix}=
\begin{bmatrix}
I&0\\
-\overline{C}^{\tp}&I
\end{bmatrix}\begin{bmatrix}
\bm{E}_{k}\\
\widehat{\overline{\bm{X}}}_{k}
\end{bmatrix},\quad\forall k\in\mathbb{Z}_{+}
\]

\noindent with which we can equivalently transform (\ref{eq039}) into
\begin{equation}\label{eq040}
\aligned
\begin{bmatrix}
\bm{E}_{k+1}\\
\begin{bmatrix}-\widetilde{\bm{E}}_{k+1}\\
\widehat{\bm{D}}_{k+1}\end{bmatrix}
\end{bmatrix}
&=\begin{bmatrix}
I-PK&-PHF\\
\overline{C}^{\tp}P_{\delta}K&\overline{A}-\overline{L}\,\overline{C}+\overline{C}^{\tp}P_{\delta}HF
\end{bmatrix}
\begin{bmatrix}
\bm{E}_{k}\\
\begin{bmatrix}-\widetilde{\bm{E}}_{k}\\
\widehat{\bm{D}}_{k}\end{bmatrix}
\end{bmatrix}
-\begin{bmatrix}I\\-\overline{C}^{\tp}\end{bmatrix}\Delta\bm{N}_{k},\quad\forall k\in\mathbb{Z}_{+}.
\endaligned
\end{equation}

\noindent From (\ref{eq040}), it is obvious that the separation principle no longer works in the presence of model uncertainties, and that the same separation result with (\ref{eq030}) holds for any ESO-based feedback controller (\ref{eq028}) if and only if $P_{\delta}=0$ (i.e., there exist no model uncertainties). Despite this issue, we can present a robust ESO-based design in the following theorem.

\begin{thm}\label{thm4}
Consider the system (\ref{eq02}) with the ESO (\ref{eq038}), and let the $k$-input be applied in the ESO-based feedback form of (\ref{eq028}). Then the closed-loop system given by (\ref{eq02}), (\ref{eq028}), and (\ref{eq038}) is robustly $k$-stable if and only if
\begin{equation}\label{eq041}
\rho\left(\begin{bmatrix}
I-PK&-PHF\\
\overline{C}^{\tp}P_{\delta}K&\overline{A}-\overline{L}\,\overline{C}+\overline{C}^{\tp}P_{\delta}HF
\end{bmatrix}\right)<1.
\end{equation}

\noindent If there exist some known matrices $\Phi_{1}\in\mathbb{R}^{p\times q}$ and $\Phi_{2}\in\mathbb{R}^{r\times m}$ and unknown matrix $\Sigma\in\mathbb{R}^{q\times r}$ such that the model uncertainty $P_{\delta}$ takes a structured form of
\begin{equation}\label{eq042}
P_{\delta}=\Phi_{1}\Sigma\Phi_{2}\quad\hbox{with}\quad\Sigma^{\tp}\Sigma\leq I
\end{equation}

\noindent then (\ref{eq041}) holds, provided there exist a positive-definite matrix $Q>0$ in the form of
\begin{equation}\label{eq043}
Q=\begin{bmatrix}
Q_{11}&(\star)\\
Q_{21}&Q_{22}
\end{bmatrix}
~\hbox{with}~
\left\{\aligned
Q_{11}&\in\mathbb{R}^{p\times p}\\
Q_{21}&\in\mathbb{R}^{2p\times p}\\
Q_{22}&\in\mathbb{R}^{2p\times2p}
\endaligned\right.
\end{equation}

\noindent and a positive scalar $\tau>0$ such that
\begin{equation}\label{eq044}
\aligned
&\left[\begin{array}{cc}
-Q_{11}&(\star)\\
-Q_{21}&-Q_{22}\\
Q_{11}(I-P_{0}K)&Q_{21}^{\tp}(\overline{A}-\overline{L}\,\overline{C})-Q_{11}P_{0}HF\\
Q_{21}(I-P_{0}K)&Q_{22}(\overline{A}-\overline{L}\,\overline{C})-Q_{21}P_{0}HF\\
\tau\Phi_{2}K&\tau\Phi_{2}HF\\
0&0
\end{array}\right.\\
&\left.\begin{array}{cccc}
(\star)&(\star)&(\star)&(\star)\\
(\star)&(\star)&(\star)&(\star)\\
-Q_{11}&(\star)&(\star)&(\star)\\
-Q_{21}&-Q_{22}&(\star)&(\star)\\
0&0&-\tau I&(\star)\\
\Phi_{1}^{\tp}\left(\overline{C}Q_{21}-Q_{11}\right)&\Phi_{1}^{\tp}\left(\overline{C}Q_{22}-Q_{21}^{\tp}\right)&0&-\tau I
\end{array}\right]<0.
\endaligned
\end{equation}
\end{thm}

\begin{proof}
By resorting to the Kalman state-space description (\ref{eq040}) with the spectral radius condition (\ref{eq041}), we can accomplish the robust $k$-stability result of this theorem in the way as that of Theorem \ref{thm2}. In addition, let us denote
\[\aligned
M
&=\begin{bmatrix}
I-PK&-PHF\\
\overline{C}^{\tp}P_{\delta}K&\overline{A}-\overline{L}\,\overline{C}+\overline{C}^{\tp}P_{\delta}HF
\end{bmatrix}\\
M_{0}
&=\begin{bmatrix}
I-P_{0}K&-P_{0}HF\\
0&\overline{A}-\overline{L}\,\overline{C}
\end{bmatrix}\\
\alpha
&=\begin{bmatrix}
\Phi_{2}K&\Phi_{2}HF&0&0
\end{bmatrix}\\
\beta
&=\begin{bmatrix}
0&0&\Phi_{1}^{\tp}\left(\overline{C}Q_{21}-Q_{11}\right)&\Phi_{1}^{\tp}\left(\overline{C}Q_{22}-Q_{21}^{\tp}\right)
\end{bmatrix}.
\endaligned\]

\noindent The use of the Schur's complement formula leads to that (\ref{eq044}) holds if and only if
\begin{equation}\label{eq045}
\begin{bmatrix}
-Q&(\star)\\
QM_{0}&-Q
\end{bmatrix}
<0
\end{equation}

\noindent and
\begin{equation}\label{eq046}
\begin{bmatrix}
-Q&(\star)\\
QM_{0}&-Q
\end{bmatrix}
+\tau\alpha^{\tp}\alpha
+\tau^{-1}\beta^{\tp}\beta
<0.
\end{equation}

\noindent Due to (\ref{eq042}) and based on \cite[Lemma 2]{mjdy:11}, there exists a positive scalar $\tau>0$ satisfying (\ref{eq046}) if and only if
\begin{equation*}\label{}
\aligned
\begin{bmatrix}
-Q&(\star)\\
QM_{0}&-Q
\end{bmatrix}
+\beta^{\tp}\Sigma\alpha
+\alpha^{\tp}\Sigma^{\tp}\beta
<0
\endaligned
\end{equation*}

\noindent which, by algebraic manipulations, can be rewritten as
\begin{equation}\label{eq047}
\begin{bmatrix}
-Q&(\star)\\
QM&-Q
\end{bmatrix}
<0.
\end{equation}

\noindent Again using the Schur's complement formula, we can leverage (\ref{eq047}) to obtain
\[
M^{\tp}QM-Q<0
\]

\noindent which, together with the Lyapunov stability criteria, guarantees $\rho(M)<1$ (namely, (\ref{eq041}) holds).
\end{proof}

\begin{rem}\label{rem06}
From the comparison with Theorem \ref{thm3}, Theorem \ref{thm4} clearly reveals that the robust $k$-stability result for ESO-based design can be generalized to overcome the effects arising from model uncertainties. This generalization applies to coping with the robust tracking problem of the plant (\ref{eq01}) under the updating law (\ref{eq037}), regardless of (\ref{eq01}) subjected to both model and external uncertainties. For Theorem \ref{thm4}, three gain matrices $\overline{L}$, $K$, and $H$ should be selected to fulfill (\ref{eq041}) despite any model uncertainty $P_{\delta}$, of which a feasible verification condition is given in (\ref{eq044}). A necessary condition of (\ref{eq044}) is shown by (\ref{eq045}) that is equivalent to $\rho\left(M_{0}\right)<1$. Since $M_{0}$ takes the upper block-triangular form, $\rho\left(M_{0}\right)<1$ holds if and only if there hold both (\ref{eq017}) and
\begin{equation}\label{eq048}
\rho\left(I-P_{0}K\right)<1.
\end{equation}

\noindent Clearly, (\ref{eq017}) always holds for some $\overline{L}$ due to the observability of $\left(\overline{A},\overline{C}\right)$, and a necessary and sufficient condition for selecting $K$ to satisfy (\ref{eq048}) is that $P_{0}$ has the full-row rank.
\end{rem}
\begin{rem}\label{rem07}
In Theorem \ref{thm4}, the robust $k$-superstability can not be developed for the closed-loop system given by (\ref{eq02}), (\ref{eq028}), and (\ref{eq038}) any longer. Since the application of the nominal ESO (\ref{eq038}) actually leads to
\begin{equation}\label{eq049}
\widetilde{\overline{\bm{X}}}_{k+1}
=\left(\overline{A}-\overline{L}\,\overline{C}\right)\widetilde{\overline{\bm{X}}}_{k}
+\overline{B}_{\delta}\overline{\bm{U}}_{k}
+\overline{\bm{D}}_{k},\quad\forall k\in\mathbb{Z}_{+}
\end{equation}

\noindent we can not obtain the properties of the observation error $\widetilde{\overline{\bm{X}}}_{k}$ in Lemma \ref{lem3} because of the presence of the model uncertainty $P_{\delta}$. To overcome this drawback, we can explore (\ref{eq049}) by employing ESO-based feedback controllers, and particularly can leverage (\ref{eq028}) to deduce
\[
\widetilde{\overline{\bm{X}}}_{k+1}
=\left(\overline{A}-\overline{L}\,\overline{C}+\overline{B}_{\delta}HF\right)\widetilde{\overline{\bm{X}}}_{k}
-\overline{B}_{\delta}K\bm{E}_{k}
-\overline{B}_{\delta}H\bm{D}_{k}
+\overline{\bm{D}}_{k}
\]

\noindent which can be integrated to obtain
\begin{equation}\label{eq050}
\aligned
\begin{bmatrix}
\bm{E}_{k+1}\\
\widetilde{\overline{\bm{X}}}_{k+1}
\end{bmatrix}
&=\begin{bmatrix}I-PK&PHF\\
-\overline{C}^{\tp}P_{\delta}K&\overline{A}-\overline{L}\,\overline{C}+\overline{C}^{\tp}P_{\delta}HF\end{bmatrix}
\begin{bmatrix}
\bm{E}_{k}\\
\widetilde{\overline{\bm{X}}}_{k}
\end{bmatrix}
-\begin{bmatrix}I-PH&0\\-\overline{C}^{\tp}P_{\delta}H&F^{\tp}\end{bmatrix}
\begin{bmatrix}\Delta\bm{N}_{k}\\\Delta^{2}\bm{N}_{k}\end{bmatrix},\quad\forall k\in\mathbb{Z}_{+}.
\endaligned
\end{equation}

\noindent Since $P_{\delta}\neq0$ and the selection candidate (\ref{eq021}) making $I-PH=0$ is not applicable any longer, (\ref{eq050}) discloses that $\Delta\bm{N}_{k}$ plays a dominant role in the robust convergence performances of both $\bm{E}_{k}$ and $\widetilde{\overline{\bm{X}}}_{k}$, rather than $\Delta^{2}\bm{N}_{k}$. Hence, the robust $k$-stability can be developed in Theorem \ref{thm4}, whereas the robust $k$-superstability can not.
\end{rem}

From the abovementioned discussions and analyses, we can see that the robust ESO-based design result of Theorem \ref{thm4} does not fully accommodate the model uncertainties. A main reason is that the nominal ESO (\ref{eq038}) focuses only on dealing with the iteration-varying external uncertainties, and thus is not capable of addressing the effect resulting from the model uncertainties. Obviously, although the nominal ESO (\ref{eq038}) is available, its $k$-state may provide helpful estimation information for only the external uncertainties, which however does not consider, and hence can not estimate, the model uncertainties.

\subsection{ESO-Based Redesign Results Against Model Uncertainties}

Next, we first redesign an ESO by taking into account model uncertainties, based on which we then establish the ESO-based design results for ILC, regardless of the simultaneous existence of iteration-varying external uncertainties. By separating $P_{\delta}\neq0$, we can rewrite (\ref{eq02}) in the form of
\begin{equation}\label{eq051}
\aligned
\bm{E}_{k+1}
&=\bm{E}_{k}+P_{0}\overline{\bm{U}}_{k}+\left(\bm{D}_{k}+P_{\delta}\overline{\bm{U}}_{k}\right)\\
&=\bm{E}_{k}+P_{0}\overline{\bm{U}}_{k}+\mathcal{D}_{k},\quad\forall k\in\mathbb{Z}_{+}
\endaligned
\end{equation}

\noindent where
\begin{equation}\label{eq052}
\mathcal{D}_{k}
=\bm{D}_{k}+P_{\delta}\overline{\bm{U}}_{k},\quad\forall k\in\mathbb{Z}_{+}.
\end{equation}

\noindent It is worth highlighting that in (\ref{eq052}), $\mathcal{D}_{k}$ collects the information for both external uncertainties and model uncertainties, unlike $\bm{D}_{k}$ in (\ref{eq03}) concerned with only the external uncertainties.

To proceed with further discussions on (\ref{eq051}), let us construct an extended Kalman state-space description as
\begin{equation}\label{eq053}
\left\{\aligned
\overline{\mathcal{X}}_{k+1}
&=\overline{A}\,\overline{\mathcal{X}}_{k}+\overline{B}_{0}\overline{\bm{U}}_{k}+\overline{\mathcal{D}}_{k}\\
\overline{\mathcal{Y}}_{k}
&=\overline{C}\,\overline{\mathcal{X}}_{k}
\endaligned,\quad\forall k\in\mathbb{Z}_{+}\right.
\end{equation}

\noindent where 
\[
\overline{\mathcal{X}}_{k}=\begin{bmatrix}\bm{E}_{k}\\\mathcal{D}_{k}\end{bmatrix}\in\mathbb{R}^{2p},\quad
\overline{\mathcal{D}}_{k}=\begin{bmatrix}0\\\Delta\mathcal{D}_{k}\end{bmatrix}\in\mathbb{R}^{2p},\quad
\overline{\mathcal{Y}}_{k}=\bm{E}_{k}.
\]

\noindent Since $\overline{A}$, $\overline{B}_{0}$, and $\overline{C}$ are available, and $\left(\overline{A},\overline{C}\right)$ is observable (see also Lemma \ref{lem2}), we consider an observer $k$-state of $\overline{\mathcal{X}}_{k}$ as
\[
\widehat{\overline{\mathcal{X}}}_{k}=\begin{bmatrix}\widehat{\bm{E}}_{k}\\\widehat{\mathcal{D}}_{k}\end{bmatrix}
\]

\noindent and based on (\ref{eq053}), we develop an ESO in a Kalman state-space form of
\begin{equation}\label{eq054}
\widehat{\overline{\mathcal{X}}}_{k+1}
=\left(\overline{A}-\overline{L}\,\overline{C}\right)\widehat{\overline{\mathcal{X}}}_{k}
+\overline{B}_{0}\overline{\bm{U}}_{k}+\overline{L}\bm{E}_{k},\quad\forall k\in\mathbb{Z}_{+}.
\end{equation}

\noindent Denote the observation error between the extended $k$-state $\overline{\mathcal{X}}_{k}$ and its observer $k$-state  $\widehat{\overline{\mathcal{X}}}_{k}$ as $\widetilde{\overline{\mathcal{X}}}_{k}=\overline{\mathcal{X}}_{k}-\widehat{\overline{\mathcal{X}}}_{k}$, or by entries,
\[\widetilde{\overline{\mathcal{X}}}_{k}
=\begin{bmatrix}\widetilde{\bm{E}}_{k}\\\widetilde{\mathcal{D}}_{k}\end{bmatrix}
=\begin{bmatrix}\bm{E}_{k}-\widehat{\bm{E}}_{k}\\\mathcal{D}_{k}-\widehat{\mathcal{D}}_{k}\end{bmatrix},\quad\forall k\in\mathbb{Z}_{+}
\]

\noindent and we can combine (\ref{eq053}) and (\ref{eq054}) to obtain
\begin{equation}\label{eq055}
\widetilde{\overline{\mathcal{X}}}_{k+1}
=\left(\overline{A}-\overline{L}\,\overline{C}\right)\widetilde{\overline{\mathcal{X}}}_{k}
+\overline{\mathcal{D}}_{k},\quad\forall k\in\mathbb{Z}_{+}
\end{equation}

\noindent where the use of (\ref{eq03}) and (\ref{eq052}) yields
\begin{equation}\label{eq056}
\aligned
\overline{\mathcal{D}}_{k}
&=F^{\tp}\left(\Delta\bm{D}_{k}+P_{\delta}\Delta\overline{\bm{U}}_{k}\right)\\
&=-F^{\tp}\left(\Delta^{2}\bm{N}_{k}+P_{\delta}\Delta^{2}\bm{U}_{k}\right),\quad\forall k\in\mathbb{Z}_{+}.
\endaligned
\end{equation}

\noindent Though (\ref{eq017}) always holds under some gain matrix $\overline{L}$, the same boundedness and superattractiveness properties of $\widetilde{\overline{\bm{X}}}_{k}$ revealed by Lemma \ref{lem3} may no longer directly apply to $\widetilde{\overline{\mathcal{X}}}_{k}$. It is because different from $\overline{\bm{D}}_{k}$ in (\ref{eq015}), $\overline{\mathcal{D}}_{k}$ in (\ref{eq055}) is tied to not only $\bm{N}_{k}$ but also $\bm{U}_{k}$ on account of the model uncertainty $P_{\delta}$, whereas there exist no prior knowledge of $\bm{U}_{k}$ that needs to be determined.

Now with the ESO (\ref{eq054}), we propose an ESO-based feedback controller instead of (\ref{eq028}) as
\begin{equation}\label{eq057}
\overline{\bm{U}}_{k}
=-K\bm{E}_{k}-H\widehat{\mathcal{D}}_{k}\\
=-K\left(\bm{E}_{k}+\overline{H}\widehat{\mathcal{D}}_{k}\right),\quad\forall k\in\mathbb{Z}_{+}
\end{equation}

\noindent where we use a specific selection $H=K\overline{H}$ for some $\overline{H}\in\mathbb{R}^{p\times p}$. Different from (\ref{eq028}), the ESO-based feedback controller (\ref{eq057}) is capable of leveraging the uniform gain matrix $K$ to transfer the feedback information brought from $\bm{E}_{k}$ and $\widehat{\mathcal{D}}_{k}$ simultaneously.

Similarly to Theorem \ref{thm4}, the following theorem develops the same robust $k$-stability result for the system (\ref{eq02}) executed under the ESO (\ref{eq054}) and the ESO-based feedback controller (\ref{eq057}).

\begin{thm}\label{thm5}
Consider the system (\ref{eq02}) with the ESO (\ref{eq054}), and let the $k$-input be applied in the ESO-based feedback form of (\ref{eq057}). Then the closed-loop system given by (\ref{eq02}), (\ref{eq054}), and (\ref{eq057}) is robustly $k$-stable if and only if the spectral radius condition (\ref{eq041}) holds, for which a sufficient condition is presented in (\ref{eq044}) when $P_{\delta}$ is structured in the form of (\ref{eq042}).
\end{thm}

\begin{proof}
For the system (\ref{eq02}) under (\ref{eq054}) and (\ref{eq057}), we develop a closed-loop state-space description as
\begin{equation}\label{eq058}
\aligned
\begin{bmatrix}
\bm{E}_{k+1}\\
\widehat{\overline{\mathcal{X}}}_{k+1}
\end{bmatrix}
&=\begin{bmatrix}
I-PK&-PHF\\
\overline{L}-\overline{B}_{0}K&\overline{A}-\overline{L}\,\overline{C}-\overline{B}_{0}HF
\end{bmatrix}
\begin{bmatrix}
\bm{E}_{k}\\
\widehat{\overline{\mathcal{X}}}_{k}
\end{bmatrix}
+\begin{bmatrix}I\\0\end{bmatrix}\bm{D}_{k},\quad\forall k\in\mathbb{Z}_{+}
\endaligned
\end{equation}

\noindent which, by the following nonsingular linear transformation:
\[
\begin{bmatrix}
\bm{E}_{k}\\
\begin{bmatrix}-\widetilde{\bm{E}}_{k}\\
\widehat{\mathcal{D}}_{k}\end{bmatrix}
\end{bmatrix}=
\begin{bmatrix}
I&0\\
-\overline{C}^{\tp}&I
\end{bmatrix}\begin{bmatrix}
\bm{E}_{k}\\
\widehat{\overline{\mathcal{X}}}_{k}
\end{bmatrix},\quad\forall k\in\mathbb{Z}_{+}
\]

\noindent can be equivalently transformed into
\begin{equation}\label{eq059}
\aligned
\begin{bmatrix}
\bm{E}_{k+1}\\
\begin{bmatrix}-\widetilde{\bm{E}}_{k+1}\\
\widehat{\mathcal{D}}_{k+1}\end{bmatrix}
\end{bmatrix}
&=\begin{bmatrix}
I-PK&-PHF\\
\overline{C}^{\tp}P_{\delta}K&\overline{A}-\overline{L}\,\overline{C}+\overline{C}^{\tp}P_{\delta}HF
\end{bmatrix}
\begin{bmatrix}
\bm{E}_{k}\\
\begin{bmatrix}-\widetilde{\bm{E}}_{k}\\
\widehat{\mathcal{D}}_{k}\end{bmatrix}
\end{bmatrix}
-\begin{bmatrix}I\\-\overline{C}^{\tp}\end{bmatrix}\Delta\bm{N}_{k},\quad\forall k\in\mathbb{Z}_{+}.
\endaligned
\end{equation}

\noindent In the same way as the proof of Theorem \ref{thm4}, we can prove this theorem by exploring (\ref{eq059}), of which the proof is omitted.
\end{proof}

With Theorem \ref{thm5}, we reveal that the redesign of the ESO (\ref{eq054}) and the ESO-based feedback controller (\ref{eq057}) also accomplishes the robust $k$-stability for the system (\ref{eq02}) as the ESO (\ref{eq038}) and the ESO-based feedback controller (\ref{eq028}) under the same condition. Furthermore, the redesign of (\ref{eq054}) and (\ref{eq057}) has advantages over the design of (\ref{eq028}) and (\ref{eq038}), especially in helping to realize the robust $k$-superstability of (\ref{eq02}). For example, as a benefit of (\ref{eq054}) and (\ref{eq057}), the (eigenvalue) separation principle can be achieved between the designs of ESO and feedback controller although the system (\ref{eq02}) is subjected to model uncertainties.

\begin{lem}\label{lem6}
For the system (\ref{eq02}), if the ESO (\ref{eq054}) and the ESO-based feedback controller (\ref{eq057}) are applied, then the syntheses of them can be guaranteed to fulfill the (eigenvalue) separation principle, where $K$, $\overline{H}$, and $\overline{L}$ can be designed separately.
\end{lem}

\begin{proof}
By employing (\ref{eq052}), we can redescribe (\ref{eq058}) as
\begin{equation}\label{eq060}
\aligned
\begin{bmatrix}
\bm{E}_{k+1}\\
\widehat{\overline{\mathcal{X}}}_{k+1}
\end{bmatrix}
&=\begin{bmatrix}
I-P_{0}K&-P_{0}K\overline{H}F\\
\overline{L}-\overline{B}_{0}K&\overline{A}-\overline{L}\,\overline{C}-\overline{B}_{0}K\overline{H}F
\end{bmatrix}
\begin{bmatrix}
\bm{E}_{k}\\
\widehat{\overline{\mathcal{X}}}_{k}
\end{bmatrix}
+\begin{bmatrix}I\\0\end{bmatrix}\mathcal{D}_{k},\quad\forall k\in\mathbb{Z}_{+}
\endaligned
\end{equation}

\noindent for which we can follow the same lines as used in (\ref{eq030}) to get
\begin{equation}\label{eq061}
\aligned
\begin{bmatrix}
I&0\\
-\overline{C}^{\tp}&I
\end{bmatrix}
&\begin{bmatrix}
I-P_{0}K&-P_{0}K\overline{H}F\\
\overline{L}-\overline{B}_{0}K&\overline{A}-\overline{L}\,\overline{C}-\overline{B}_{0}K\overline{H}F
\end{bmatrix}
\begin{bmatrix}
I&0\\
-\overline{C}^{\tp}&I
\end{bmatrix}^{-1}
=\begin{bmatrix}
I-P_{0}K
&-P_{0}K\overline{H}F\\
0
&\overline{A}-\overline{L}\,\overline{C}
\end{bmatrix}.
\endaligned
\end{equation}

\noindent With (\ref{eq061}), the separation principle holds between the syntheses of (\ref{eq054}) and (\ref{eq057}) for the system (\ref{eq02}).
\end{proof}

In spite of the satisfaction of separation principle in Lemma \ref{lem6}, we can obtain from (\ref{eq060}) that $\mathcal{D}_{k}$ plays the role as its driving input and greatly affects the $k$-stability performance. It is worth emphasizing, however, that $\mathcal{D}_{k}$ is heavily dependent upon the $k$-input $\overline{\bm{U}}_{k}$ based on (\ref{eq052}), whereas there is no prior knowledge of $\overline{\bm{U}}_{k}$, such as boundedness or convergence. To overcome this drawback, we propose a useful lemma to develop boundedness properties of $\bm{U}_{k}$, and consequently of $\overline{\bm{U}}_{k}$.

\begin{lem}\label{lem7}
Consider the system (\ref{eq02}) under the ESO (\ref{eq054}) and the ESO-based feedback controller (\ref{eq057}), and assume that both the spectral radius conditions (\ref{eq017}) and (\ref{eq048}) are fulfilled. Then $\bm{U}_{k}$ is bounded such that $\left\|\bm{U}_{k}\right\|\leq\beta_{\bm{U}}$ is satisfied for some finite bound $\beta_{\bm{U}}\geq0$ if and only if
\begin{equation}\label{eq062}
\rho\left(\begin{bmatrix}
I-PK&\overline{H}F\\
-\overline{L}P_{\delta}K&\overline{A}-\overline{L}\,\overline{C}
\end{bmatrix}\right)<1
\end{equation}

\noindent and, moreover, there can be found some finite constants $\gamma_{1}\geq0$ and $\gamma_{2}\geq0$ and some constant $\lambda\in[0,1)$ such that
\begin{equation}\label{eq063}
\aligned
\left\|\Delta\bm{U}_{k}\right\|
&\leq\gamma_{1}\beta_{\Delta\bm{N}}+\gamma_{2}\lambda^{k}\\
\left\|\Delta^{2}\bm{U}_{k}\right\|
&\leq\gamma_{1}\beta_{\Delta^{2}\bm{N}}+\gamma_{2}\lambda^{k}
\endaligned,\quad\forall k\in\mathbb{Z}_{+}
\end{equation}

\noindent and some finite constant $\gamma_{3}\geq0$ such that
\begin{equation}\label{eq064}
\aligned
\limsup_{k\to\infty}\left\|\Delta\bm{U}_{k}\right\|
&\leq\gamma_{3}\beta_{\Delta\bm{N}}^{ess}\\
\limsup_{k\to\infty}\left\|\Delta^{2}\bm{U}_{k}\right\|
&\leq\gamma_{3}\beta_{\Delta^{2}\bm{N}}^{ess}.
\endaligned
\end{equation}

\noindent If $P_{\delta}$ is structured in the form of (\ref{eq042}), then (\ref{eq017}), (\ref{eq048}), and (\ref{eq062}) can be achieved, provided there exist a positive-definite matrix $Q>0$ of the form (\ref{eq043}) and a positive scalar $\tau>0$ satisfying
\begin{equation}\label{eq065}
\aligned
&\left[\begin{array}{cc}
-Q_{11}&(\star)\\
-Q_{21}&-Q_{22}\\
Q_{11}(I-P_{0}K)&Q_{11}\overline{H}F+Q_{21}^{\tp}(\overline{A}-\overline{L}\,\overline{C})\\
Q_{21}(I-P_{0}K)&Q_{21}\overline{H}F+Q_{22}(\overline{A}-\overline{L}\,\overline{C})\\
\tau\Phi_{2}K&0\\
0&0
\end{array}\right.\\
&\left.\begin{array}{cccc}
(\star)&(\star)&(\star)&(\star)\\
(\star)&(\star)&(\star)&(\star)\\
-Q_{11}&(\star)&(\star)&(\star)\\
-Q_{21}&-Q_{22}&(\star)&(\star)\\
0&0&-\tau I&(\star)\\
\Phi_{1}^{\tp}\left(-Q_{11}-\overline{L}^{\tp}Q_{21}\right)&\Phi_{1}^{\tp}\left(-Q_{21}^{\tp}-\overline{L}^{\tp}Q_{22}\right)&0&-\tau I
\end{array}\right]\\
&~~<0.
\endaligned
\end{equation}
\end{lem}

\begin{proof}
Three steps are included to prove this lemma.

{\it Step i):} From (\ref{eq01}), (\ref{eq03}), and (\ref{eq057}), we can arrive at
\begin{equation}\label{eq066}
\aligned
\bm{U}_{k+1}
&=\bm{U}_{k}-\overline{\bm{U}}_{k}\\
&=\bm{U}_{k}+K\bm{E}_{k}+K\overline{H}\widehat{\mathcal{D}}_{k}\\
&=\left(I-KP\right)\bm{U}_{k}+K\overline{H}\widehat{\mathcal{D}}_{k}
+K\left(\bm{Y}_{d}-\bm{N}_{k}\right),\quad\forall k\in\mathbb{Z}_{+}.
\endaligned
\end{equation}

\noindent Based on the spectral radius condition (\ref{eq048}), $P_{0}K$ is nonsingular and, thus, $P_{0}$ and $K$ have full-row and full-column rank, respectively. Without loss of generality, we denote $P_{0}=\left[P_{01}~P_{02}\right]$ with $P_{01}\in\mathbb{R}^{p\times p}$ and $P_{02}\in\mathbb{R}^{p\times(m-p)}$, for which $P_{01}$ is assumed to be nonsingular (if not, this can be realized through the elementary transformation). Correspondingly, let us denote $K=\left[K_{1}^{\tp}~K_{2}^{\tp}\right]^{\tp}$ with $K_{1}\in\mathbb{R}^{p\times p}$ and $K_{2}\in\mathbb{R}^{(m-p)\times p}$. As a consequence, we can present a nonsingular transformation of (\ref{eq066}) as
\[
\bm{U}_{k}^{\ast}
=\Omega^{-1}\bm{U}_{k}
~\hbox{with}~
\bm{U}_{k}^{\ast}
=\begin{bmatrix}
\bm{U}_{1,k}^{\ast}\\
\bm{U}_{2,k}^{\ast}
\end{bmatrix}
\]

\noindent where $\bm{U}_{1,k}^{\ast}\in\mathbb{R}^{p}$, $\bm{U}_{2,k}^{\ast}\in\mathbb{R}^{m-p}$, and
\[
\Omega^{-1}
=\begin{bmatrix}
P_{01}&P_{02}\\
-K_{2}\left(P_{0}K\right)^{-1}P_{01}&I-K_{2}\left(P_{0}K\right)^{-1}P_{02}
\end{bmatrix}.
\]

\noindent In fact, we can validate
\[
\Omega
=\begin{bmatrix}
K_{1}\left(P_{0}K\right)^{-1}&-P_{01}^{-1}P_{02}\\
K_{2}\left(P_{0}K\right)^{-1}&I
\end{bmatrix},\quad
\Omega^{-1}K
=\begin{bmatrix}P_{0}K\\0\end{bmatrix},\quad
P_{0}\Omega
=\begin{bmatrix}I&0\end{bmatrix}
\]

\noindent and consequently,
\begin{equation*}\label{}
\aligned
\bm{U}_{k+1}^{\ast}
&=\left(I-\Omega^{-1}KP\Omega\right)\bm{U}_{k}^{\ast}+\Omega^{-1}K\overline{H}\widehat{\mathcal{D}}_{k}
+\Omega^{-1}K\left(\bm{Y}_{d}-\bm{N}_{k}\right)\\
&=\begin{bmatrix}
P_{0}K\left(I-PK\right)\left(P_{0}K\right)^{-1}
&-P_{0}KP_{\delta}\begin{bmatrix}-P_{01}^{-1}P_{02}\\I\end{bmatrix}\\
0&I
\end{bmatrix}\bm{U}_{k}^{\ast}
+\begin{bmatrix}P_{0}K\overline{H}\widehat{\mathcal{D}}_{k}\\0\end{bmatrix}
+\begin{bmatrix}P_{0}K\left(\bm{Y}_{d}-\bm{N}_{k}\right)\\0\end{bmatrix}
\endaligned
\end{equation*}

\noindent with which we equivalently have
\begin{equation}\label{eq067}
\aligned
\bm{U}_{1,k+1}^{\ast}
&=P_{0}K\left(I-PK\right)\left(P_{0}K\right)^{-1}\bm{U}_{1,k}^{\ast}+P_{0}K\overline{H}\widehat{\mathcal{D}}_{k}\\
&~~~+P_{0}K\left(\bm{Y}_{d}-\bm{N}_{k}\right)
-P_{0}KP_{\delta}\begin{bmatrix}-P_{01}^{-1}P_{02}\\I\end{bmatrix}\bm{U}_{2,k}^{\ast},\quad\forall k\in\mathbb{Z}_{+}
\endaligned
\end{equation}

\noindent and
\begin{equation}\label{eq068}
\bm{U}_{2,k+1}^{\ast}
=\bm{U}_{2,k}^{\ast},\quad\forall k\in\mathbb{Z}_{+}.
\end{equation}

\noindent Clearly, (\ref{eq068}) is equivalent to that $\bm{U}_{2,k}^{\ast}$ is iteration-invariant and bounded, namely,
\begin{equation}\label{eq069}
\aligned
\bm{U}_{2,k}^{\ast}
&\equiv\bm{U}_{2,0}^{\ast}\\
&=\begin{bmatrix}
-K_{2}\left(P_{0}K\right)^{-1}P_{01}&I-K_{2}\left(P_{0}K\right)^{-1}P_{02}
\end{bmatrix}\bm{U}_{0},\quad\forall k\in\mathbb{Z}_{+}.
\endaligned
\end{equation}

\noindent For (\ref{eq067}), let us denote
\[
\bm{U}_{1,k}^{\diamond}
=\left(P_{0}K\right)^{-1}\bm{U}_{1,k}^{\ast},\quad\forall k\in\mathbb{Z}_{+}
\]

\noindent and then by combining (\ref{eq069}), we can deduce
\begin{equation}\label{eq070}
\aligned
\bm{U}_{1,k+1}^{\diamond}
&=\left(I-PK\right)\bm{U}_{1,k}^{\diamond}
+\overline{H}\widehat{\mathcal{D}}_{k}
+\bm{Y}_{d}-\bm{N}_{k}
-P_{\delta}\begin{bmatrix}-P_{01}^{-1}P_{02}\\I\end{bmatrix}\bm{U}_{2,0}^{\ast}\\
&=\left(I-PK\right)\bm{U}_{1,k}^{\diamond}
+\overline{H}F\widehat{\overline{\mathcal{X}}}_{k}
+\bm{\theta}_{k},\quad\forall k\in\mathbb{Z}_{+}
\endaligned
\end{equation}

\noindent where
\begin{equation}\label{eq071}
\bm{\theta}_{k}
=\bm{Y}_{d}-\bm{N}_{k}
-P_{\delta}\begin{bmatrix}-P_{01}^{-1}P_{02}\\I\end{bmatrix}\bm{U}_{2,0}^{\ast},~~\forall k\in\mathbb{Z}_{+}.
\end{equation}

\noindent We can verify with (\ref{eq071}) that $\bm{\theta}_{k}$ is bounded, namely, $\left\|\bm{\theta}_{k}\right\|\leq\beta_{\bm{\theta}}$ for some bound $\beta_{\bm{\theta}}\geq0$ given by
\[\aligned
\beta_{\bm{\theta}}
&=\left\|\bm{Y}_{d}\right\|
+\beta_{\bm{N}}
+\beta_{\delta}\left\|\begin{bmatrix}-P_{01}^{-1}P_{02}\\I\end{bmatrix}\right\|
\left\|\begin{bmatrix}
-K_{2}\left(P_{0}K\right)^{-1}P_{01}&I-K_{2}\left(P_{0}K\right)^{-1}P_{02}
\end{bmatrix}\right\|
\left\|\bm{U}_{0}\right\|.
\endaligned\]

{\it Step ii):} Let us combine (\ref{eq054}) and (\ref{eq057}) to deduce
\begin{equation*}\label{}
\aligned
\widehat{\overline{\mathcal{X}}}_{k+1}
&=\left(\overline{A}-\overline{L}\,\overline{C}\right)\widehat{\overline{\mathcal{X}}}_{k}
-\overline{B}_{0}K\left(\bm{E}_{k}+\overline{H}\widehat{\mathcal{D}}_{k}\right)
+\overline{L}\bm{E}_{k}\\
&=\left(\overline{A}-\overline{L}\,\overline{C}-\overline{B}_{0}K\overline{H}F\right)\widehat{\overline{\mathcal{X}}}_{k}
+\left(\overline{L}-\overline{B}_{0}K\right)\bm{E}_{k}\\
&=\left(\overline{A}-\overline{L}\,\overline{C}-\overline{B}_{0}K\overline{H}F\right)\widehat{\overline{\mathcal{X}}}_{k}
+\left(\overline{B}_{0}K-\overline{L}\right)P\bm{U}_{k}
+\left(\overline{L}-\overline{B}_{0}K\right)\left(\bm{Y}_{d}-\bm{N}_{k}\right),\quad\forall k\in\mathbb{Z}_{+}
\endaligned
\end{equation*}

\noindent which, together with $\bm{U}_{k}^{\ast}=\Omega^{-1}\bm{U}_{k}$, leads to
\begin{equation}\label{eq072}
\aligned
\widehat{\overline{\mathcal{X}}}_{k+1}
&=\left(\overline{A}-\overline{L}\,\overline{C}-\overline{B}_{0}K\overline{H}F\right)\widehat{\overline{\mathcal{X}}}_{k}
+\left(\overline{B}_{0}K-\overline{L}\right)P\Omega\bm{U}_{k}^{\ast}
+\left(\overline{L}-\overline{B}_{0}K\right)\left(\bm{Y}_{d}-\bm{N}_{k}\right),\quad\forall k\in\mathbb{Z}_{+}.
\endaligned
\end{equation}

\noindent If we notice the fact of (\ref{eq069}), then we can validate
\begin{equation}\label{eq073}
\aligned
P\Omega\bm{U}_{k}^{\ast}
&=P_{0}\Omega\bm{U}_{k}^{\ast}
+P_{\delta}\begin{bmatrix}
K_{1}\left(P_{0}K\right)^{-1}&-P_{01}^{-1}P_{02}\\
K_{2}\left(P_{0}K\right)^{-1}&I
\end{bmatrix}\bm{U}_{k}^{\ast}\\
&=\left[I+P_{\delta}K\left(P_{0}K\right)^{-1}\right]\bm{U}_{1,k}^{\ast}
+P_{\delta}\begin{bmatrix}
-P_{01}^{-1}P_{02}\\
I
\end{bmatrix}\bm{U}_{2,0}^{\ast}\\
&=PK\bm{U}_{1,k}^{\diamond}
+P_{\delta}\begin{bmatrix}
-P_{01}^{-1}P_{02}\\
I
\end{bmatrix}\bm{U}_{2,0}^{\ast},\quad\forall k\in\mathbb{Z}_{+}.
\endaligned
\end{equation}

\noindent By substituting (\ref{eq073}) into (\ref{eq072}) and adopting (\ref{eq071}), we can obtain
\begin{equation}\label{eq074}
\aligned
\widehat{\overline{\mathcal{X}}}_{k+1}
&=\left(\overline{A}-\overline{L}\,\overline{C}-\overline{B}_{0}K\overline{H}F\right)\widehat{\overline{\mathcal{X}}}_{k}
+\left(\overline{B}_{0}K-\overline{L}\right)PK\bm{U}_{1,k}^{\diamond}
+\left(\overline{L}-\overline{B}_{0}K\right)\bm{\theta}_{k},\quad\forall k\in\mathbb{Z}_{+}.
\endaligned
\end{equation}

\noindent Obviously, we can write (\ref{eq070}) and (\ref{eq074}) in a compact form of
\begin{equation}\label{eq075}
\aligned
\begin{bmatrix}
\bm{U}_{1,k+1}^{\diamond}\\
\widehat{\overline{\mathcal{X}}}_{k+1}
\end{bmatrix}
&=\begin{bmatrix}
I-PK&\overline{H}F\\
\left(\overline{B}_{0}K-\overline{L}\right)PK&\overline{A}-\overline{L}\,\overline{C}-\overline{B}_{0}K\overline{H}F
\end{bmatrix}
\begin{bmatrix}
\bm{U}_{1,k}^{\diamond}\\
\widehat{\overline{\mathcal{X}}}_{k}
\end{bmatrix}
+\begin{bmatrix}I\\\overline{L}-\overline{B}_{0}K\end{bmatrix}\bm{\theta}_{k},\quad\forall k\in\mathbb{Z}_{+}.
\endaligned
\end{equation}

\noindent For (\ref{eq075}), we have
\begin{equation}\label{eq076}
\aligned
\begin{bmatrix}
I&0\\
\overline{B}_{0}K&I
\end{bmatrix}
\begin{bmatrix}
I-PK&\overline{H}F\\
\left(\overline{B}_{0}K-\overline{L}\right)PK&\overline{A}-\overline{L}\,\overline{C}-\overline{B}_{0}K\overline{H}F
\end{bmatrix}
\begin{bmatrix}
I&0\\
\overline{B}_{0}K&I
\end{bmatrix}^{-1}
&=\begin{bmatrix}
I-PK&\overline{H}F\\
-\overline{L}P_{\delta}K&\overline{A}-\overline{L}\,\overline{C}
\end{bmatrix}\\
&\triangleq\mathcal{M}.
\endaligned
\end{equation}

\noindent Based on (\ref{eq076}) and with the boundedness of $\bm{\theta}_{k}$ obtained in the {\it Step i)}, we can derive from (\ref{eq075}) that if $\rho(\mathcal{M})<1$, namely, (\ref{eq062}) holds, then both $\bm{U}_{1,k}^{\diamond}$ and $\widehat{\overline{\mathcal{X}}}_{k}$ are bounded, namely, $\left\|\bm{U}_{1,k}^{\diamond}\right\|\leq\beta_{\bm{U}^{\diamond}}$ and $\left\|\widehat{\overline{\mathcal{X}}}_{k}\right\|\leq\beta_{\widehat{\overline{\mathcal{X}}}}$ for some bounds $\beta_{\bm{U}^{\diamond}}\geq0$ and $\beta_{\widehat{\overline{\mathcal{X}}}}\geq0$. Consequently, we can verify the boundedness of $\bm{U}_{k}$ such that
\[\aligned
\left\|\bm{U}_{k}\right\|
&\leq\left\|\Omega\right\|\left(\left\|\bm{U}_{1,k}^{\ast}\right\|+\left\|\bm{U}_{2,k}^{\ast}\right\|\right)\\
&\leq\left\|\Omega\right\|\left(\left\|P_{0}K\right\|\left\|\bm{U}_{1,k}^{\diamond}\right\|+\left\|\bm{U}_{2,0}^{\ast}\right\|\right)\\
&\leq\left\|\Omega\right\|\big[\left\|P_{0}K\right\|\beta_{\bm{U}^{\diamond}}
+\left\|\begin{bmatrix}-K_{2}\left(P_{0}K\right)^{-1}P_{01}&I-K_{2}\left(P_{0}K\right)^{-1}P_{02}\end{bmatrix}\right\|
\left\|\bm{U}_{0}\right\|\big]\\
&\triangleq\beta_{\bm{U}},\quad\forall k\in\mathbb{Z}_{+}.
\endaligned
\]

\noindent On the contrary, if $\bm{U}_{k}$ is bounded such that $\left\|\bm{U}_{k}\right\|\leq\beta_{\bm{U}}$ holds for some bound $\beta_{\bm{U}}\geq0$, then
\[\aligned
\left\|\bm{U}_{1,k}^{\diamond}\right\|
&\leq\left\|\left(P_{0}K\right)^{-1}\begin{bmatrix}I&0\end{bmatrix}\Omega^{-1}\right\|\left\|\bm{U}_{k}\right\|\\
&\leq\left\|\left(P_{0}K\right)^{-1}\begin{bmatrix}I&0\end{bmatrix}\Omega^{-1}\right\|\beta_{\bm{U}}\\
&\triangleq\beta_{\bm{U}^{\diamond}},\quad\forall k\in\mathbb{Z}_{+}
\endaligned\]

\noindent and
\[\aligned
\left\|\bm{E}_{k}\right\|
&\leq\left\|\bm{Y}_{d}\right\|
+\left(\left\|P_{0}\right\|+\left\|P_{\delta}\right\|\right)\left\|\bm{U}_{k}\right\|
+\left\|\bm{N}_{k}\right\|\\
&\leq\left\|\bm{Y}_{d}\right\|
+\left(\left\|P_{0}\right\|+\beta_{\delta}\right)\beta_{\bm{U}}
+\beta_{\bm{N}}\\
&\triangleq\beta_{\bm{E}},\quad\forall k\in\mathbb{Z}_{+}.
\endaligned
\]

\noindent Based on the boundedness of both $\bm{E}_{k}$ and $\bm{U}_{k}$, we analyze (\ref{eq054}) under the spectral radius condition (\ref{eq017}) and can establish the boundedness of $\widehat{\overline{\mathcal{X}}}_{k}$ such that $\left\|\widehat{\overline{\mathcal{X}}}_{k}\right\|\leq\beta_{\widehat{\overline{\mathcal{X}}}}$ for some bound $\beta_{\widehat{\overline{\mathcal{X}}}}\geq0$. This, together with the boundedness of $\bm{U}_{1,k}^{\diamond}$ and $\bm{\theta}_{k}$, implies that we can conclude $\rho(\mathcal{M})<1$ from (\ref{eq075}) by resorting to the relation (\ref{eq076}).

In addition, the use of (\ref{eq071}) yields $\Delta\bm{\theta}_{k}=-\Delta\bm{N}_{k}$ and $\Delta^{2}\bm{\theta}_{k}=-\Delta^{2}\bm{N}_{k}$, and thus we can employ (\ref{eq075}) to deduce
\begin{equation}\label{eq077}
\aligned
\begin{bmatrix}
\Delta\bm{U}_{1,k+1}^{\diamond}\\
\Delta\widehat{\overline{\mathcal{X}}}_{k+1}
\end{bmatrix}
&=\begin{bmatrix}
I-PK&\overline{H}F\\
\left(\overline{B}_{0}K-\overline{L}\right)PK&\overline{A}-\overline{L}\,\overline{C}-\overline{B}_{0}K\overline{H}F
\end{bmatrix}
\begin{bmatrix}
\Delta\bm{U}_{1,k}^{\diamond}\\
\Delta\widehat{\overline{\mathcal{X}}}_{k}
\end{bmatrix}
-\begin{bmatrix}I\\\overline{L}-\overline{B}_{0}K\end{bmatrix}\Delta\bm{N}_{k},\quad\forall k\in\mathbb{Z}_{+}
\endaligned
\end{equation}

\noindent and
\begin{equation}\label{eq078}
\aligned
\begin{bmatrix}
\Delta^{2}\bm{U}_{1,k+1}^{\diamond}\\
\Delta^{2}\widehat{\overline{\mathcal{X}}}_{k+1}
\end{bmatrix}
&=\begin{bmatrix}
I-PK&\overline{H}F\\
\left(\overline{B}_{0}K-\overline{L}\right)PK&\overline{A}-\overline{L}\,\overline{C}-\overline{B}_{0}K\overline{H}F
\end{bmatrix}
\begin{bmatrix}
\Delta^{2}\bm{U}_{1,k}^{\diamond}\\
\Delta^{2}\widehat{\overline{\mathcal{X}}}_{k}
\end{bmatrix}
-\begin{bmatrix}I\\\overline{L}-\overline{B}_{0}K\end{bmatrix}\Delta^{2}\bm{N}_{k},\quad\forall k\in\mathbb{Z}_{+}.
\endaligned
\end{equation}

\noindent With (\ref{eq062}), we benefit from (\ref{eq076}) to derive that there exists some induced matrix norm to satisfy (see also \cite[Lemma 5.6.10]{hj:85})
\[
\left\|\begin{bmatrix}
I-PK&\overline{H}F\\
\left(\overline{B}_{0}K-\overline{L}\right)PK&\overline{A}-\overline{L}\,\overline{C}-\overline{B}_{0}K\overline{H}F
\end{bmatrix}\right\|\triangleq\lambda<1
\]

\noindent and consequently, we can employ (\ref{eq077}) and (\ref{eq078}) to arrive at
\begin{equation}\label{eq079}
\aligned
\left\|\Delta\bm{U}_{1,k}^{\diamond}\right\|
&\leq\lambda^{k}\left\|\begin{bmatrix}
\Delta\bm{U}_{1,0}^{\diamond}\\
\Delta\widehat{\overline{\mathcal{X}}}_{0}
\end{bmatrix}\right\|
+\frac{\Ds\left\|\begin{bmatrix}I\\\overline{L}-\overline{B}_{0}K\end{bmatrix}\right\|}
{\Ds1-\lambda}\beta_{\Delta\bm{N}}\\
\left\|\Delta^{2}\bm{U}_{1,k}^{\diamond}\right\|
&\leq\lambda^{k}\left\|\begin{bmatrix}
\Delta^{2}\bm{U}_{1,0}^{\diamond}\\
\Delta^{2}\widehat{\overline{\mathcal{X}}}_{0}
\end{bmatrix}\right\|
+\frac{\Ds\left\|\begin{bmatrix}I\\\overline{L}-\overline{B}_{0}K\end{bmatrix}\right\|}
{\Ds1-\lambda}\beta_{\Delta^{2}\bm{N}}
\endaligned,\quad\forall k\in\mathbb{Z}_{+}.
\end{equation}

\noindent We can also leverage (\ref{eq077}) and (\ref{eq078}) to deduce
\begin{equation}\label{eq080}
\aligned
\limsup_{k\to\infty}\left\|\Delta\bm{U}_{1,k}^{\diamond}\right\|
&\leq\frac{\Ds\left\|\begin{bmatrix}I\\\overline{L}-\overline{B}_{0}K\end{bmatrix}\right\|}
{\Ds1-\lambda}\beta_{\Delta\bm{N}}^{ess}\\
\limsup_{k\to\infty}\left\|\Delta^{2}\bm{U}_{1,k}^{\diamond}\right\|
&\leq\frac{\Ds\left\|\begin{bmatrix}I\\\overline{L}-\overline{B}_{0}K\end{bmatrix}\right\|}
{\Ds1-\lambda}\beta_{\Delta^{2}\bm{N}}^{ess}.
\endaligned
\end{equation}

\noindent Due to
\[\aligned
\bm{U}_{k}
&=\Omega\bm{U}_{k}^{\ast}\\
&=\begin{bmatrix}
K_{1}\left(P_{0}K\right)^{-1}\\
K_{2}\left(P_{0}K\right)^{-1}
\end{bmatrix}\bm{U}_{1,k}^{\ast}
+\begin{bmatrix}
-P_{01}^{-1}P_{02}\\
I
\end{bmatrix}\bm{U}_{2,k}^{\ast}\\
&=K\bm{U}_{1,k}^{\diamond}
+\begin{bmatrix}
-P_{01}^{-1}P_{02}\\
I
\end{bmatrix}\bm{U}_{2,0}^{\ast},\quad\forall k\in\mathbb{Z}_{+}
\endaligned\]

\noindent we can deduce
\begin{equation}\label{eq081}
\Delta\bm{U}_{k}
=K\Delta\bm{U}_{1,k}^{\diamond},\quad\forall k\in\mathbb{Z}_{+}
\end{equation}

\noindent and
\begin{equation}\label{eq082}
\Delta^{2}\bm{U}_{k}
=K\Delta^{2}\bm{U}_{1,k}^{\diamond},\quad\forall k\in\mathbb{Z}_{+}.
\end{equation}

\noindent Based on (\ref{eq081}) and (\ref{eq082}), we can obtain (\ref{eq063}) (respectively, (\ref{eq064})) from (\ref{eq079}) (respectively, (\ref{eq080})).

{\it Step iii):} Let us denote
\[\aligned
\mathcal{M}_{0}
&=\begin{bmatrix}
I-P_{0}K&\overline{H}F\\
0&\overline{A}-\overline{L}\,\overline{C}
\end{bmatrix}\\
\tilde{\alpha}
&=\begin{bmatrix}
\Phi_{2}K&0&0&0
\end{bmatrix}\\
\tilde{\beta}&=\begin{bmatrix}
0&0&\Phi_{1}^{\tp}\left(-Q_{11}-\overline{L}^{\tp}Q_{21}\right)&\Phi_{1}^{\tp}\left(-Q_{21}^{\tp}-\overline{L}^{\tp}Q_{22}\right)
\end{bmatrix}.
\endaligned\]

\noindent Then with the use of the Schur's complement lemma, we know that (\ref{eq065}) holds if and only if
\begin{equation}\label{eq083}
\begin{bmatrix}
-Q&(\star)\\
Q\mathcal{M}_{0}&-Q
\end{bmatrix}<0
\end{equation}

\noindent and
\begin{equation}\label{eq084}
\aligned
\begin{bmatrix}
-Q&(\star)\\
Q\mathcal{M}_{0}&-Q
\end{bmatrix}
+\tau\tilde{\alpha}^{\tp}\tilde{\alpha}
+\tau^{-1}\tilde{\beta}^{\tp}\tilde{\beta}
<0.
\endaligned
\end{equation}

\noindent From \cite[Lemma 2]{mjdy:11}, the satisfaction of (\ref{eq084}) for some positive scalar $\tau>0$ leads to
\begin{equation*}\label{}
\aligned
\begin{bmatrix}
-Q&(\star)\\
Q\mathcal{M}_{0}&-Q
\end{bmatrix}
+\tilde{\beta}^{\tp}\Sigma\tilde{\alpha}
+\tilde{\alpha}^{\tp}\Sigma^{\tp}\tilde{\beta}
&=\begin{bmatrix}
-Q&(\star)\\
Q\mathcal{M}&-Q
\end{bmatrix}\\
&<0
\endaligned
\end{equation*}

\noindent which implies $\rho(\mathcal{M})<1$ regardless of any model uncertainties satisfying (\ref{eq042}) according to the Lyapunov stability criteria. For the same reason, we benefit from (\ref{eq083}) to derive that $\rho(\mathcal{M}_{0})<1$ holds. A consequence of the upper block-triangular structure of $\mathcal{M}_{0}$ is that both (\ref{eq017}) and (\ref{eq048}) are satisfied. Thus, we complete the proof of Lemma \ref{lem7}.
\end{proof}

\begin{rem}\label{rem08}
From Lemma \ref{lem7}, we can gain the boundedness of all signals of our interest, such as $\bm{E}_{k}$, $\overline{\bm{U}}_{k}$, $\bm{Y}_{k}$, $\mathcal{D}_{k}$, $\overline{\mathcal{D}}_{k}$, $\overline{\mathcal{X}}_{k}$, and $\widehat{\overline{\mathcal{X}}}_{k}$. This actually benefits from the redesign of the ESO (\ref{eq054}) and the ESO-based feedback controller (\ref{eq057}). In particular, (\ref{eq063}) and (\ref{eq064}) can provide estimations for $\overline{\bm{U}}_{k}$ and $\Delta\overline{\bm{U}}_{k}$, respectively, owing to $\overline{\bm{U}}_{k}=-\Delta\bm{U}_{k}$.
%
%
\end{rem}

With Lemma \ref{lem7}, we are in position to propose the ESO-based controller design of the system (\ref{eq02}) in the following theorem.

\begin{thm}\label{thm6}
Consider the system (\ref{eq02}) with the ESO (\ref{eq054}), and let the $k$-input be designed in the ESO-based feedback form of (\ref{eq057}) under the spectral radius condition (\ref{eq062}). Then the closed-loop system described by (\ref{eq02}), (\ref{eq054}), and (\ref{eq057}) is robustly $k$-stable if and only if both the spectral radius conditions (\ref{eq017}) and (\ref{eq048}) hold. Further, when adopting the selection candidate of $\overline{H}$ as
\begin{equation}\label{eq085}
\overline{H}=\left(P_{0}K\right)^{-1}
\end{equation}

\noindent the robust $k$-superstability can be accomplished for the closed-loop system described by (\ref{eq02}), (\ref{eq054}), and (\ref{eq057}) if and only if both the spectral radius conditions (\ref{eq017}) and (\ref{eq048}) are satisfied, where there can be found some class $\mathcal{K}_{\infty}$ function $\chi$ such that
\begin{equation}\label{eq086}
\limsup_{k\to\infty}\left\|\bm{E}_{k}\right\|
\leq\chi\left(\limsup_{k\to\infty}\left\|\widetilde{\mathcal{D}}_{k}\right\|\right).
\end{equation}

\noindent In particular, the linear matrix inequality condition (\ref{eq065}) can be utilized to verify the satisfaction of (\ref{eq062}) when $P_{\delta}$ is structured in the form of (\ref{eq042}).
\end{thm}

\begin{proof}
By substituting (\ref{eq057}) into (\ref{eq051}), we can obtain
\begin{equation}\label{eq087}
\aligned
\bm{E}_{k+1}
&=\left(I-P_{0}K\right)\bm{E}_{k}
-P_{0}K\overline{H}\widehat{\mathcal{D}}_{k}+\mathcal{D}_{k}\\
&=\left(I-P_{0}K\right)\bm{E}_{k}
+P_{0}K\overline{H}F\widetilde{\overline{\mathcal{X}}}_{k}
+\left(I-P_{0}K\overline{H}\right)\mathcal{D}_{k},~\forall k\in\mathbb{Z}_{+}
\endaligned
\end{equation}

\noindent which, together with (\ref{eq03}), (\ref{eq052}), (\ref{eq055}), and (\ref{eq056}), results in
\begin{equation}\label{eq088}
\aligned
\begin{bmatrix}
\bm{E}_{k+1}\\
\widetilde{\overline{\mathcal{X}}}_{k+1}
\end{bmatrix}
&=\begin{bmatrix}
I-P_{0}K&P_{0}K\overline{H}F\\
0&\overline{A}-\overline{L}\,\overline{C}
\end{bmatrix}
\begin{bmatrix}
\bm{E}_{k}\\
\widetilde{\overline{\mathcal{X}}}_{k}
\end{bmatrix}
+\begin{bmatrix}I-P_{0}K\overline{H}&0\\0&I\end{bmatrix}
\begin{bmatrix}\mathcal{D}_{k}\\\overline{\mathcal{D}}_{k}\end{bmatrix}\\
&=\begin{bmatrix}
I-P_{0}K&P_{0}K\overline{H}F\\
0&\overline{A}-\overline{L}\,\overline{C}
\end{bmatrix}
\begin{bmatrix}
\bm{E}_{k}\\
\widetilde{\overline{\mathcal{X}}}_{k}
\end{bmatrix}
-\begin{bmatrix}I-P_{0}K\overline{H}&0\\0&F^{\tp}\end{bmatrix}
\begin{bmatrix}\Delta\bm{N}_{k}\\\Delta^{2}\bm{N}_{k}\end{bmatrix}\\
&~~~-\begin{bmatrix}\left(I-P_{0}K\overline{H}\right)P_{\delta}&0\\0&F^{\tp}P_{\delta}\end{bmatrix}
\begin{bmatrix}\Delta\bm{U}_{k}\\\Delta^{2}\bm{U}_{k}\end{bmatrix},\quad\forall k\in\mathbb{Z}_{+}.
\endaligned
\end{equation}

\noindent Based on (\ref{eq088}) and for the same reason as the proof of Theorem \ref{thm2}, we can conclude that the robust $k$-stability holds only if
\[
\rho\left(\begin{bmatrix}
I-P_{0}K&P_{0}K\overline{H}F\\
0&\overline{A}-\overline{L}\,\overline{C}
\end{bmatrix}\right)<1
\]

\noindent which is equivalent to (\ref{eq017}) and (\ref{eq048}). Conversely, two spectral radius conditions ensure the existence of some induced matrix norm to satisfy
\[
\left\|\begin{bmatrix}
I-P_{0}K&P_{0}K\overline{H}F\\
0&\overline{A}-\overline{L}\,\overline{C}
\end{bmatrix}\right\|\triangleq\hat{\lambda}<1.
\]

\noindent With (\ref{eq063}) and (\ref{eq064}), we denote $\lambda_{\max}=\max\left\{\lambda,\hat{\lambda}\right\}$ and $\lambda_{\min}=\min\left\{\lambda,\hat{\lambda}\right\}$, and hence the use of (\ref{eq088}) leads to
\begin{equation}\label{eq093}
\aligned
\left\|\bm{E}_{k}\right\|
&\leq\beta_{0}\hat{\lambda}^{k}
+\frac{\Ds3\beta_{1}}{\Ds1-\hat{\lambda}}\beta_{\Delta\bm{N}}
+3\beta_{2}\sum_{i=0}^{k-1}\hat{\lambda}^{k-1-i}\left(\gamma_{1}\beta_{\Delta\bm{N}}+\gamma_{2}\lambda^{i}\right)\\
&\leq\beta_{0}\lambda_{\max}^{k}
+\frac{\Ds3\left(\beta_{1}+\gamma_{1}\beta_{2}\right)}{\Ds1-\hat{\lambda}}\beta_{\Delta\bm{N}}
+3\gamma_{2}\beta_{2}\sum_{i=0}^{k-1}\lambda_{\max}^{k-1-i}\lambda_{\min}^{i}\\
&\leq\left(\beta_{0}+\frac{\Ds3\gamma_{2}\beta_{2}}{\Ds\lambda_{\max}-\lambda_{\min}}\right)\lambda_{\max}^{k}
+\frac{\Ds3\left(\beta_{1}+\gamma_{1}\beta_{2}\right)}{\Ds1-\hat{\lambda}}\beta_{\Delta\bm{N}}
\endaligned
\end{equation}

\noindent and
\begin{equation}\label{eq094}
\aligned
\limsup_{k\to\infty}\left\|\bm{E}_{k}\right\|
&\leq\frac{\Ds3\left(\beta_{1}+\gamma_{3}\beta_{2}\right)}{\Ds1-\hat{\lambda}}\beta_{\Delta\bm{N}}^{ess}
\endaligned
\end{equation}

\noindent where
\[\aligned
\beta_{0}=\left\|\begin{bmatrix}\bm{E}_{0}\\\widetilde{\overline{\mathcal{X}}}_{0}\end{bmatrix}\right\|,\quad
\beta_{1}=\left\|\begin{bmatrix}I-P_{0}K\overline{H}&0\\0&F^{\tp}\end{bmatrix}\right\|,\quad
\beta_{2}=\left\|\begin{bmatrix}\left(I-P_{0}K\overline{H}\right)P_{\delta}&0\\0&F^{\tp}P_{\delta}\end{bmatrix}\right\|.
\endaligned\]

\noindent Clearly, (\ref{eq093}) and (\ref{eq094}) can ensure the robust $k$-stability of (\ref{eq088}) according to Definition \ref{def01}.

When considering the selection of $\overline{H}$ in (\ref{eq085}), (\ref{eq088}) becomes
\begin{equation}\label{eq089}
\aligned
\begin{bmatrix}
\bm{E}_{k+1}\\
\widetilde{\overline{\mathcal{X}}}_{k+1}
\end{bmatrix}
&=\begin{bmatrix}
I-P_{0}K&F\\
0&\overline{A}-\overline{L}\,\overline{C}
\end{bmatrix}
\begin{bmatrix}
\bm{E}_{k}\\
\widetilde{\overline{\mathcal{X}}}_{k}
\end{bmatrix}
-\begin{bmatrix}0\\F^{\tp}\end{bmatrix}\Delta^{2}\bm{N}_{k}
-\begin{bmatrix}0\\F^{\tp}P_{\delta}\end{bmatrix}\Delta^{2}\bm{U}_{k},\quad\forall k\in\mathbb{Z}_{+}.
\endaligned
\end{equation}

\noindent Based on (\ref{eq089}) and with (\ref{eq063}) and (\ref{eq064}), we can follow the same steps used in the above proof of robust $k$-stability to conclude that (\ref{eq017}) and (\ref{eq048}) are necessary and sufficient for the robust $k$-superstability. In particular, it follows from (\ref{eq089}) that due to $\widetilde{\mathcal{D}}_{k}=F\widetilde{\overline{\mathcal{X}}}_{k}$,
\begin{equation}\label{eq090}
\bm{E}_{k+1}
=\left(I-P_{0}K\right)\bm{E}_{k}
+\widetilde{\mathcal{D}}_{k},\quad\forall k\in\mathbb{Z}_{+}.
\end{equation}

\noindent By again resorting to (\ref{eq048}), we can gain (\ref{eq086}) based on exploring (\ref{eq090}) in the same way as the derivation of (\ref{eq031}).

Additionally, the satisfaction of (\ref{eq062}) based on (\ref{eq065}) has been examined in Lemma \ref{lem7}. The proofs of all results in this theorem are thus completed.
\end{proof}

\begin{rem}\label{rem09}
In Theorem \ref{thm6}, the design of the ESO (\ref{eq054}) and the ESO-based feedback controller (\ref{eq057}) is enabled to reach not only the robust $k$-stability but also the robust $k$-superstability of the system (\ref{eq02}), regardless of the simultaneous presence of both model and external uncertainties. It is particularly possible for the ESO-based design to render the tracking error continuously dependent on the observation error with respect to the variation of the iteration-varying uncertainty even in the presence of the model uncertainty. By Theorem \ref{thm6}, the tracking error decreases to vanish, especially when the variation of the iteration-varying uncertainty quasi-disappears. Since $\mathcal{D}_{k}$ involves all uncertainty information, regardless of from the plant model or the external disturbance, (\ref{eq086}) and (\ref{eq090}) disclose that the tracking error only depends on the observation error of all uncertainties. This gives the possibility to achieve high-precision tracking tasks through the proposed ESO-based design method.
\end{rem}
\begin{rem}\label{rem011}
In the same way as the proof of Theorem \ref{thm6}, we can also estimate the observation error $\widetilde{\overline{\mathcal{X}}}_{k}$. More specifically, the substitution of (\ref{eq056}) into (\ref{eq055}) arrives at
\begin{equation*}\label{}
\widetilde{\overline{\mathcal{X}}}_{k+1}
=\left(\overline{A}-\overline{L}\,\overline{C}\right)\widetilde{\overline{\mathcal{X}}}_{k}
-F^{\tp}\Delta^{2}\bm{N}_{k}
-F^{\tp}P_{\delta}\Delta^{2}\bm{U}_{k},\quad\forall k\in\mathbb{Z}_{+}.
\end{equation*}

\noindent Let $\tilde{\lambda}=\left\|\overline{A}-\overline{L}\,\overline{C}\right\|<1$, and then using (\ref{eq063}) in Lemma \ref{lem7} yields
\[
\left\|\widetilde{\overline{\mathcal{X}}}_{k}\right\|
\leq\left(\left\|\widetilde{\overline{\mathcal{X}}}_{0}\right\|
+\frac{\Ds\gamma_{2}\beta_{\delta}}{\Ds\lambda_{\max}-\lambda_{\min}}\right)\lambda_{\max}^{k}
+\frac{\Ds1+\gamma_{1}\beta_{\delta}}{\Ds1-\tilde{\lambda}}\beta_{\Delta^{2}\bm{N}}
\]

\noindent where $\lambda_{\max}=\max\left\{\lambda,\tilde{\lambda}\right\}$ and $\lambda_{\min}=\min\left\{\lambda,\tilde{\lambda}\right\}$. Similarly, the use of (\ref{eq064}) in Lemma \ref{lem7} leads to
\begin{equation*}\label{}
\limsup_{k\to\infty}\left\|\widetilde{\overline{\mathcal{X}}}_{k}\right\|
\leq\frac{\Ds1+\gamma_{3}\beta_{\delta}}{\Ds1-\tilde{\lambda}}\beta_{\Delta^{2}\bm{N}}^{ess}.
\end{equation*}

\noindent Namely, the boundedness and superattractiveness properties of $\widetilde{\overline{\bm{X}}}_{k}$ in Lemma \ref{lem3} can be induced for $\widetilde{\overline{\mathcal{X}}}_{k}$.
\end{rem}

As an application for the robust tracking of the plant (\ref{eq01}), we can equivalently develop an updating law from the ESO-based feedback controller (\ref{eq057}) as
\begin{equation}\label{eq091}
\aligned
\bm{U}_{k+1}
&=\bm{U}_{k}+K\bm{E}_{k}+H\widehat{\mathcal{D}}_{k}\\
&=\bm{U}_{k}+K\left(\bm{E}_{k}+\overline{H}\widehat{\mathcal{D}}_{k}\right),\quad\forall k\in\mathbb{Z}_{+}.
\endaligned
\end{equation}

\noindent From Theorem \ref{thm6}, it follows that the ESO-based design of (\ref{eq054}) and (\ref{eq091}) works robustly for the plant (\ref{eq01}) to realize the tracking of any desired output target, regardless of unknown model and external uncertainties. A perfect tracking objective particularly can be realized only when the variation of the iteration-varying uncertainty is driven to quasi-disappear. These reflect that (\ref{eq091}) not only maintains the same robust tracking performances with (\ref{eq037}) but also greatly improves them to effectively work against the model uncertainties.

\section{Model-Free Design Without Plant Knowledge}\label{sec5}

In this section, we attempt to explore the Kalman state-space design methods of data-driven learning for the model-free case of the plant (\ref{eq01}), i.e., $P_{0}=0$. Since there exists no prior nominal model knowledge of the plant (\ref{eq01}), those previously established results in Sections \ref{sec3} and \ref{sec4} are no longer applicable. In spite of this crucial issue, we incorporate the idea of Theorem \ref{thm6} into the ESO-based controller analysis and design in the model-free case of (\ref{eq01}), and construct some full-row rank matrix $\widetilde{P}_{0}\in\mathbb{R}^{p\times m}$ to overcome the difficulty caused by $P_{0}=0$. As a consequence, we can also determine some gain matrix $K$ such that
\begin{equation}\label{eq095}
\rho\left(I-\widetilde{P}_{0}K\right)<1.
\end{equation}

If we denote $\widetilde{P}_{\delta}=P_{\delta}-\widetilde{P}_{0}$, then the system (\ref{eq02}) can read as
\begin{equation}\label{eq096}
\bm{E}_{k+1}
=\bm{E}_{k}+\widetilde{P}_{0}\overline{\bm{U}}_{k}+\mathbb{D}_{k},\quad\forall k\in\mathbb{Z}_{+}
\end{equation}

\noindent where
\begin{equation}\label{eq097}
\mathbb{D}_{k}
=\bm{D}_{k}+\widetilde{P}_{\delta}\overline{\bm{U}}_{k},\quad\forall k\in\mathbb{Z}_{+}.
\end{equation}

\noindent From (\ref{eq096}), we construct an extended $k$-state and its estimation state, respectively, as
\[
\overline{\mathbb{X}}_{k}=\begin{bmatrix}\bm{E}_{k}\\\mathbb{D}_{k}\end{bmatrix}\in\mathbb{R}^{2p},\quad
\widehat{\overline{\mathbb{X}}}_{k}=\begin{bmatrix}\widehat{\bm{E}}_{k}\\\widehat{\mathbb{D}}_{k}\end{bmatrix}\in\mathbb{R}^{2p},\quad\forall k\in\mathbb{Z}_{+}
\]

\noindent and then propose an ESO in a Kalman state-space form of
\begin{equation}\label{eq098}
\widehat{\overline{\mathbb{X}}}_{k+1}
=\left(\overline{A}-\overline{L}\,\overline{C}\right)\widehat{\overline{\mathbb{X}}}_{k}
+\overline{\widetilde{B}}_{0}\overline{\bm{U}}_{k}+\overline{L}\bm{E}_{k},\quad\forall k\in\mathbb{Z}_{+}
\end{equation}

\noindent where
\[
\overline{\widetilde{B}}_{0}
=\overline{C}^{\tp}\widetilde{P}_{0}
=\begin{bmatrix}
\widetilde{P}_{0}\\
0\\
\end{bmatrix}.
\]

\noindent Note that (\ref{eq098}) is an available ESO, and we can always provide it with the condition (\ref{eq017}). Similarly to (\ref{eq057}), the use of the ESO (\ref{eq098}) leads to an ESO-based feedback controller as
\begin{equation}\label{eq099}
\overline{\bm{U}}_{k}
=-K\left(\bm{E}_{k}+\overline{H}\,\widehat{\mathbb{D}}_{k}\right),\quad\forall k\in\mathbb{Z}_{+}.
\end{equation}

In addition, let $\widetilde{P}_{0}$ be selected such that $\widetilde{P}_{\delta}$ can be structured in a form of
\begin{equation}\label{eq100}
\widetilde{P}_{\delta}=\widetilde{\Phi}_{1}\widetilde{\Sigma}\widetilde{\Phi}_{2}\quad\hbox{with}\quad \widetilde{\Sigma}^{\tp}\widetilde{\Sigma}\leq I
\end{equation}

\noindent for some known matrices $\widetilde{\Phi}_{1}\in\mathbb{R}^{p\times q}$ and $\widetilde{\Phi}_{2}\in\mathbb{R}^{r\times m}$ and some unknown matrix $\widetilde{\Sigma}\in\mathbb{R}^{q\times r}$. Correspondingly, we propose a test condition given by
\begin{equation}\label{eq101}
\aligned
&\left[\begin{array}{cc}
-Q_{11}&(\star)\\
-Q_{21}&-Q_{22}\\
Q_{11}(I-\widetilde{P}_{0}K)&Q_{11}\overline{H}F+Q_{21}^{\tp}(\overline{A}-\overline{L}\,\overline{C})\\
Q_{21}(I-\widetilde{P}_{0}K)&Q_{21}\overline{H}F+Q_{22}(\overline{A}-\overline{L}\,\overline{C})\\
\tau \widetilde{\Phi}_{2}K&0\\
0&0
\end{array}\right.\\
&\left.\begin{array}{cccc}
(\star)&(\star)&(\star)&(\star)\\
(\star)&(\star)&(\star)&(\star)\\
-Q_{11}&(\star)&(\star)&(\star)\\
-Q_{21}&-Q_{22}&(\star)&(\star)\\
0&0&-\tau I&(\star)\\
\widetilde{\Phi}_{1}^{\tp}\left(-Q_{11}-\overline{L}^{\tp}Q_{21}\right)&\widetilde{\Phi}_{1}^{\tp}\left(-Q_{21}^{\tp}-\overline{L}^{\tp}Q_{22}\right)&0&-\tau I
\end{array}\right]\\
&~~<0
\endaligned
\end{equation}

\noindent for some positive-definite matrix $Q>0$ in the form of (\ref{eq043}) and some positive scalar $\tau>0$. Now, we can develop the following ESO-based design result in the model-free case.

\begin{thm}\label{thm7}
Consider the system (\ref{eq02}) with the ESO (\ref{eq098}), and let the $k$-input be designed in the ESO-based feedback form of (\ref{eq099}), and $\widetilde{P}_{0}$, $K$, and $\overline{L}$ be selected to fulfill the spectral radius condition as
\begin{equation}\label{eq102}
\aligned
\rho\left(\begin{bmatrix}
I-P_{\delta}K&\overline{H}F\\
\overline{L}\left(\widetilde{P}_{0}-P_{\delta}\right)K&\overline{A}-\overline{L}\,\overline{C}
\end{bmatrix}\right)<1.
\endaligned
\end{equation}

\noindent Then the closed-loop system given in the form of (\ref{eq02}), (\ref{eq098}), and (\ref{eq099}) is robustly $k$-stable if and only if both the spectral radius conditions (\ref{eq017}) and (\ref{eq095}) hold. Furthermore, when adopting the selection candidate of $\overline{H}$ as
\begin{equation}\label{eq103}
\overline{H}=\left(\widetilde{P}_{0}K\right)^{-1}
\end{equation}

\noindent the robust $k$-superstability can be accomplished for the closed-loop system described by (\ref{eq02}), (\ref{eq098}), and (\ref{eq099}) if and only if both the spectral radius conditions (\ref{eq017}) and (\ref{eq095}) are satisfied, where there can be found some class $\mathcal{K}_{\infty}$ function $\chi$ such that
\begin{equation}\label{eq104}
\limsup_{k\to\infty}\left\|\bm{E}_{k}\right\|
\leq\chi\left(\limsup_{k\to\infty}\left\|\mathbb{D}_{k}-\widehat{\mathbb{D}}_{k}\right\|\right).
\end{equation}

\noindent Particularly, the linear matrix inequality condition (\ref{eq101}) can be applied to validate the satisfaction of (\ref{eq102}) when $\widetilde{P}_{0}$ is selected such that $\widetilde{P}_{\delta}$ is structured in the form of (\ref{eq100}).
\end{thm}

\begin{proof}
This proof can be derived in the same way as that of Theorem \ref{thm6}.
\end{proof}

\begin{rem}\label{rem12}
For the extreme case when $P_{\delta}=\widetilde{P}_{0}$, Theorem \ref{thm7} always holds and provides the same result with Theorem \ref{thm3}. For the other extreme case when $P_{\delta}=0$ (namely, $P=0$), the plant (\ref{eq01}) collapses into $\bm{Y}_{k}=\bm{N}_{k}$, $\forall k\in\mathbb{Z}_{+}$ that does not have control efforts. It is hence trivial that the plant (\ref{eq01}) can not achieve any tracking task. This is also consistent with Theorem \ref{thm7}. Actually, the spectral radius condition (\ref{eq102}) is never satisfied, regardless of any selections of $\widetilde{P}_{0}$, $K$, and $\overline{L}$, due to
\begin{equation*}\label{}
\aligned
\begin{bmatrix}
\widetilde{P}_{0}K&0\\
-\overline{C}^{\tp}\widetilde{P}_{0}K&I
\end{bmatrix}
&\begin{bmatrix}
I&\overline{H}F\\
\overline{L}\widetilde{P}_{0}K&\overline{A}-\overline{L}\,\overline{C}
\end{bmatrix}
\begin{bmatrix}
\widetilde{P}_{0}K&0\\
-\overline{C}^{\tp}\widetilde{P}_{0}K&I
\end{bmatrix}^{-1}
=\begin{bmatrix}
I&\widetilde{P}_{0}K\overline{H}F\\
0&\overline{A}-\overline{L}\,\overline{C}-\overline{C}^{\tp}\widetilde{P}_{0}K\overline{H}F
\end{bmatrix}.
\endaligned
\end{equation*}

\end{rem}

Based on the robust stability results of Theorems \ref{thm3}, \ref{thm6}, and \ref{thm7}, we can arrive at a step-by-step strategy to perform the tracking task of the plant (\ref{eq01}) with any desired target even in the absence of nominal model knowledge. In particular, the use of Theorem \ref{thm7} leads to an ESO-based data-driven learning algorithm for the plant (\ref{eq01}) in the model-free case, which is given by (\ref{eq098}) and
\begin{equation}\label{eq105}
\bm{U}_{k+1}
=\bm{U}_{k}+K\left(\bm{E}_{k}+\overline{H}\,\widehat{\mathbb{D}}_{k}\right),\quad\forall k\in\mathbb{Z}_{+}.
\end{equation}

\noindent When using the gain selection (\ref{eq103}) in the updating law (\ref{eq105}), we can decrease the tracking error and ensure it to only depend on the observation error for the unknown information involved in the plant (\ref{eq01}), as shown by (\ref{eq104}). This is thanks to leveraging the idea of ESO for the design of data-driven learning.

\section{Applications and Verifications}\label{sec6}

\subsection{Applications to ILC}

Consider ILC of a system, evolving with respect to both an infinite iteration axis $k \in \mathbb{Z}_+$ and a finite time axis $t \in \mathbb{Z}_{T}\triangleq\{0,1,\cdots,T\}$, as
\begin{equation}\label{eq106}
\left\{\aligned
x_k(t+1)&=Ax_k(t)+Bu_k(t)+w_k(t)\\
y_k(t)&=Cx_k(t)+v_k(t)
\endaligned\right.
\end{equation}

\noindent where
\begin{itemize}
\item
$x_k(t)\in\mathbb{R}^{n_{s}}$--state with the initial value: $x_k(0), \forall k\in \mathbb{Z}_+$;

\item
$u_k(t)\in\mathbb{R}^{n_{i}}$--input with the initial value: $u_0(t), \forall t \in \mathbb{Z}_{T-1}$;

\item
$y_k(t)\in\mathbb{R}^{n_{o}}$--output measured for specified tracking tasks;

\item
$w_k(t)\in\mathbb{R}^{n_{s}}$, $v_k(t)\in\mathbb{R}^{n_{o}}$--state and output disturbances;

\item
$A$, $B$, $C$--system matrices with appropriate dimensions.
\end{itemize}

\noindent Of interest in ILC for the system (\ref{eq106}) is to achieve the output tracking of any desired reference $y_{d}(t)$, namely, $\lim_{k\to\infty}y_{k}(t)=y_{d}(t)$, $\forall t=1,2,\cdots,T$. To this end, let $CB$ be of full-row rank.

By the lifting technique (see, e.g., \cite{bta:06,acm:07}), we define
\[
\aligned
\bm{U}_{k}
&=\left[u_{k}^{\tp}(0),u_{k}^{\tp}(1),\cdots,u_{k}^{\tp}(T-1)\right]\\
\bm{Y}_{k}
&=\left[y_{k}^{\tp}(1),y_{k}^{\tp}(2),\cdots,y_{k}^{\tp}(T)\right]
\endaligned
\]

\noindent and let $\bm{W}_{k}$ and $\bm{E}_{k}$, $\bm{V}_{k}$ and $\bm{Y}_{d}$ be defined in the same ways as $\bm{U}_{k}$ and $\bm{Y}_{k}$, respectively. Then we can reformulate the system (\ref{eq106}) into the framework of (\ref{eq01}) by denoting $\bm{N}_{k}=Q\bm{W}_{k}+\bm{V}_{k}+Sx_{k}(0)$, where
\begin{equation}\label{eq107}
\aligned
P&=\begin{bmatrix}
    CB & 0 & \cdots & 0 \\
    CAB & CB & \ddots & \vdots \\
    \vdots & \vdots & \ddots & 0 \\
    CA^{T-1}B & CA^{T-2}B & \cdots & CB
  \end{bmatrix},\quad
Q=\begin{bmatrix}
    C & 0 & \cdots & 0 \\
    CA & C & \ddots & \vdots \\
    \vdots & \vdots & \ddots & 0 \\
    CA^{T-1} & CA^{T-2} & \cdots & C
  \end{bmatrix},\quad
  S=\begin{bmatrix}
    CA \\
    CA^{2} \\
    \vdots \\
    CA^{T}
  \end{bmatrix}.
\endaligned
\end{equation}

\noindent Thus, the ESO-based design result of Theorem \ref{thm7} is applicable to ILC of the system (\ref{eq106}). If, for $i$, $j=1$, $2$, $\cdots$, $N$, we denote
\[
\aligned
K&=\left[K_{ij}\right]~\hbox{with}~K_{ij}\in\mathbb{R}^{n_{i}\times n_{o}},\quad
\overline{H}=\left[\overline{H}_{ij}\right]~\hbox{with}~\overline{H}_{ij}\in\mathbb{R}^{n_{o}\times n_{o}}\\
\widehat{\mathbb{D}}_{k}&=\left[\widehat{d}_{k}^{\tp}(1),\widehat{d}_{k}^{\tp}(2),\cdots,\widehat{d}_{k}^{\tp}(N)\right]^{\tp}
~\hbox{with}~\widehat{d}_{k}^{\tp}(j)\in\mathbb{R}^{n_{o}}
\endaligned
\]

\noindent then we can use (\ref{eq105}) to obtain an updating law for (\ref{eq106}) as
\begin{equation}\label{eq108}
\aligned
u_{k+1}(t)&=u_{k}(t)+\sum_{j=1}^{N}K_{t+1,j}\Big[e_{k}(j)
+\sum_{j=1}^{N}\overline{H}_{t+1,j}\widehat{d}_{k}(j)\Big],\quad\forall t\in\mathbb{Z}_{N-1},\forall k\in\mathbb{Z}_{+}.
\endaligned
\end{equation}

\noindent Dfferent from traditional ILC only using the tracking error, the updating law (\ref{eq108}) incorporates the ESO-based information.

\subsection{Simulation Verifications}

For illustration, we consider the system (\ref{eq01}) with an unknown plant $P$ resulting from ILC of the system (\ref{eq106}) that is subjected to big uncertainties in the elements of the system matrices $A$, $B$, and $C$. To be specific,  there exists a maximum $30\%$ uncertainty in each element of $A$, $B$, and $C$ vs. the estimations (\cite{mjdy:11,amc:07}):
\[
A_{0}=\begin{bmatrix}0.72&0&0\\
1&-1.04&-0.81\\
0&0.81&0\end{bmatrix},\quad
B_{0}=\begin{bmatrix}1\\0\\0\end{bmatrix},\quad
C_{0}=\begin{bmatrix}1\\-0.98\\-1.09\end{bmatrix}^{\tp}.
\]

\noindent Despite this issue, we focus on realizing the output tracking of the desired target $\bm{Y}_{d}=[y_{d}(1)$, $y_{d}(2),\cdots,y_{d}(T)]^{\tp}$ induced from $y_{d}(t)=\sin\left(8t/T\right)$. For the iteration-varying uncertainty $\bm{N}_{k}$, we consider it of $T$ identical entries given by $\sum_{i=0}^{k}\sin\left(i/200\right)/(i+1)^{0.5}$. Here, we adopt $T=20$ without loss of generality.

To implement our ESO-based data-driven learning algorithm, we construct a full-row rank matrix $\widetilde{P}_{0}\in\mathbb{R}^{20\times 20}$ as
\[
\widetilde{P}_{0}
=\begin{bmatrix}
1& 0 & \cdots & \cdots & \cdots & 0 \\
-0.5& 1 & \ddots& \ddots & \ddots & \vdots \\
-0.25& -0.5 & \ddots & \ddots & \ddots & \vdots \\
0& -0.25 & \ddots & \ddots & \ddots & \vdots \\
\vdots& \ddots & \ddots & \ddots & \ddots & 0 \\
0& \cdots & 0 & -0.25 & -0.5 & 1
\end{bmatrix}
\]

\noindent which has a different matrix structure from $P$ in (\ref{eq107}). For the ESO (\ref{eq098}), we use the zero initial estimation state (i.e., $\widehat{\overline{\mathbb{X}}}_{0}=0$), and select $L_1=0.9I$ and $L_2=0.1I$ to fulfill the condition (\ref{eq017}). Hence, we can apply the ESO-based updating law (\ref{eq105}), where we adopt the zero initial input (i.e., $\bm{U}_{0}=0$) and select the gain matrix as $K=0.5\widetilde{P}_{0}^{-1}$ based on the condition (\ref{eq095}). From (\ref{eq103}), we can further determine $\overline{H}=2I$.

\begin{figure}[!t]
  \centering
  \includegraphics[width=0.6\hsize]{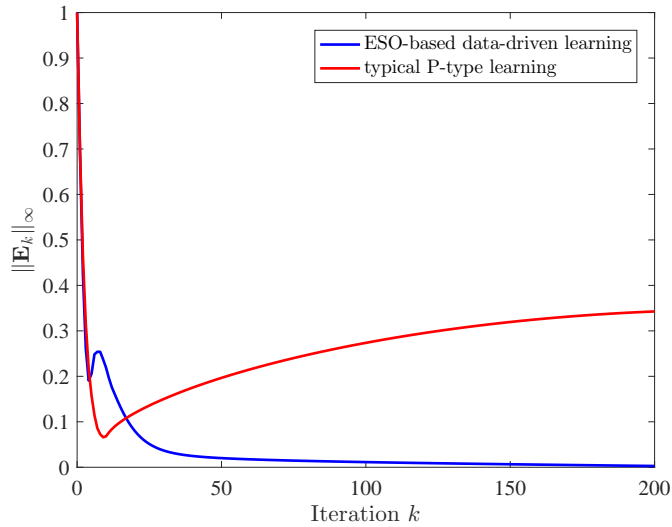}
  \caption{Iteration performances of ESO-based data-driven learning vs. typical P-type learning.}\label{fig1}
\end{figure}

To demonstrate the performances for ESO-based data-driven learning, we plot the iteration process evaluated by $\|\bm{E}_k\|_{\infty}$, and simultaneously depict the processes for both updating laws (\ref{eq09}) and (\ref{eq105}) in Fig. \ref{fig1} for comparison. This figure clearly verifies that the ESO-based data-driven learning law (\ref{eq105}) is effective in overcoming effects of unknown plant models and iteration-varying uncertainties. However, by contrast, the typical P-type updating law (\ref{eq09}) no longer works effectively. Thus, we validate the theoretical results of Theorem \ref{thm7} for data-driven learning of model-free systems in spite of iteration-varying uncertainties.

\section{Conclusions}\label{sec7}

In this paper, a data-driven learning problem with the aim to achieve the output tracking of any desired target for model-free systems has been discussed, to which an ESO-based controller design approach has been incorporated. A three-step problem-solving strategy has been adopted, which is realized by making a problem transformation to treat two robust stability problems in the Kalman state-space framework. It provides a perspective to exploit data-driven learning by benefiting from model-based control methods, regardless of model-free plants. Of particular interest is the establishment of ESO-based feedback controllers to address the effects caused by iteration-varying uncertainties. In addition, our ESO-based data-driven learning approach has been introduced to robust ILC systems in the presence of both model uncertainties and iteration-varying external disturbances and initial shifts. We have also validated the application results to ILC via illustrative simulations.

\section*{Acknowledgement}\label{}

The author would like to thank Dr. Jingyao Zhang, Beihang University, for his useful discussions and helps on implementing the simulations, and Mr. Zirong Guo, Beihang University, for his helpful discussions.

\bibliographystyle{ieeetr}

\end{document}